\documentclass[11pt,a4paper]{article}
\usepackage[lmargin=1.0in,rmargin=1.0in,bottom=1.0in,top=1.0in,twoside=False]{geometry}
\usepackage[czech, english]{babel}
\usepackage[shortlabels,]{enumitem}
\usepackage{fontaxes}
\usepackage{trimspaces}
\usepackage{nccfoots}
\usepackage{setspace}
\usepackage{inconsolata}
\usepackage{libertine}
\usepackage{amsmath,amssymb,amsthm}
\usepackage{mathtools}
\usepackage{comment}
\usepackage[all,defaultlines=4]{nowidow}
\usepackage{microtype}
\usepackage{mathrsfs}
\usepackage{mathtools}
\usepackage{marvosym}
\usepackage{thm-restate}
\usepackage{tikz}
\usetikzlibrary{calc}
\usetikzlibrary{shapes}
\definecolor{darkgreen}{rgb}{10,117,28}
\usepackage{wrapfig}
\usepackage{diagbox}
\usepackage{stackengine}

\definecolor{blue}{rgb}{0.1,0.2,0.5}
\definecolor{brown}{rgb}{0.6,0.6,0.2}
\usepackage{cleveref}

\crefformat{page}{#2page~#1#3}%
\Crefformat{page}{#2Page~#1#3}%
\crefformat{equation}{#2(#1)#3}%
\Crefformat{equation}{#2(#1)#3}%
\crefformat{figure}{#2Figure~#1#3}%
\Crefformat{figure}{#2Figure~#1#3}%
\crefformat{section}{#2Section~#1#3}
\Crefformat{section}{#2Section~#1#3}
\crefformat{chapter}{#2Chapter~#1#3}
\Crefformat{chapter}{#2Chapter~#1#3}
\crefformat{chapter*}{#2Chapter~#1#3}
\Crefformat{chapter*}{#2Chapter~#1#3}
\crefformat{part}{#2Part~#1#3}
\Crefformat{part}{#2Part~#1#3}
\crefformat{enumi}{#2(#1)#3}
\Crefformat{enumi}{#2(#1)#3}

\newenvironment{absolutelynopagebreak}
  {\par\nobreak\vfil\penalty0\vfilneg
   \vtop\bgroup}
  {\par\xdef\tpd{\the\prevdepth}\egroup
   \prevdepth=\tpd}
   
\newtheorem{theorem}{Theorem}[section]
\crefformat{theorem}{#2Theorem~#1#3}
\Crefformat{theorem}{#2Theorem~#1#3}
\newtheorem{lemma}{Lemma}[section]
\crefformat{lemma}{#2Lemma~#1#3}
\Crefformat{lemma}{#2Lemma~#1#3}
\newtheorem{conjecture}{Conjecture}[section]
\crefformat{conjecture}{#2Conjecture~#1#3}
\Crefformat{conjecture}{#2Conjecture~#1#3}
\newtheorem*{zclaim}{Claim}

\newcommand{\newtheoremwithcrefformat}[2]{%
  \newtheorem{#1}{#2}[section]%
  \crefformat{#1}{##2\MakeUppercase#1~##1##3}%
  \Crefformat{#1}{##2\MakeUppercase#1~##1##3}%
}

\newtheoremwithcrefformat{proposition}{Proposition}
\newtheoremwithcrefformat{observation}{Observation}
\newtheoremwithcrefformat{corollary}{Corollary}
\newtheoremwithcrefformat{claim}{Claim}
\newtheoremwithcrefformat{example}{Example}
\newtheoremwithcrefformat{remark}{Remark}
\newtheoremwithcrefformat{definition}{Definition}

\newenvironment{claimproof}{%
\begin{proof}[Proof of the claim]}{%
	\end{proof}
    }

\newcommand{\App}{{\Large{\Cutright}}{}}

\newcommand{\xxx}{}

\newcommand{\checkforA}[1]{}
\newcommand{\checkforB}[1]{}
\newenvironment{Alemma}{\begin{lemma}}{\end{lemma}}

\newenvironment{proofEnd}{\begin{proof}}{\end{proof}}
\makeatletter
\def\ifenv#1{
   \def\@tempa{#1}%
   \ifx\@tempa\@currenvir
      \expandafter\@firstoftwo
    \else
      \expandafter\@secondoftwo
   \fi
}
\let\wfs@comment@comment\comment
\let\comment\@undefined
\newcommand{\untoto}{\let\toto\@undefined}
\usepackage{changes}
\let\wfs@changes@comment\comment
\let\comment\@undefined

\newcommand\comment{%
    \ifthenelse{\equal{\@currenvir}{comment}}
    {\wfs@comment@comment}
    {\wfs@changes@comment}%
}
\makeatother

\usepackage{anyfontsize}
\newcommand{\LL}{\mbox{\L}}

\renewcommand{\top}{\textsf{top}}

\newcommand{\Cc}{\mathscr{C}}
\newcommand{\Dd}{\mathscr{D}}
\newcommand{\Ee}{\mathscr{E}}
\renewcommand{\phi}{\varphi}

\renewcommand{\leq}{\leqslant}
\renewcommand{\geq}{\geqslant}

\newcommand{\tup}[1]{\bar{#1}}
\newcommand{\void}{\makebox[8pt][c]{\fullmoon}}

\newcommand{\Tr}{\mathsf{T}}
\newcommand{\Pp}{\mathcal{P}}
\newcommand{\Qq}{\mathcal{Q}}
\newcommand{\Ff}{\mathcal{F}}
\newcommand{\Ss}{\mathcal{S}}
\newcommand{\Blocks}{\mathsf{Blocks}}
\newcommand{\Types}{\mathsf{Types}}
\newcommand{\pth}{\mathsf{path}}
\newcommand{\Tf}{\mathfrak{T}}

\newcommand{\SReach}{\mathrm{SReach}}
\newcommand{\scol}{\mathrm{scol}}

\newcommand{\enumref}[1]{\ref{#1}}
\newcommand{\ERCagreement}{
\vspace{-6pt}

\noindent
{\begin{minipage}[t]{.75\textwidth}\small This paper is a part of projects that have received funding from the European Research Council (ERC) under the European Union's Horizon 2020 research and innovation programme (grant agreements No 810115 -- {\sc Dynasnet}, and No 677651 -- {\sc Total}). \end{minipage}\hfill\begin{minipage}[t]{.22\textwidth}\raisebox{-42pt}{\includegraphics[width=\textwidth]{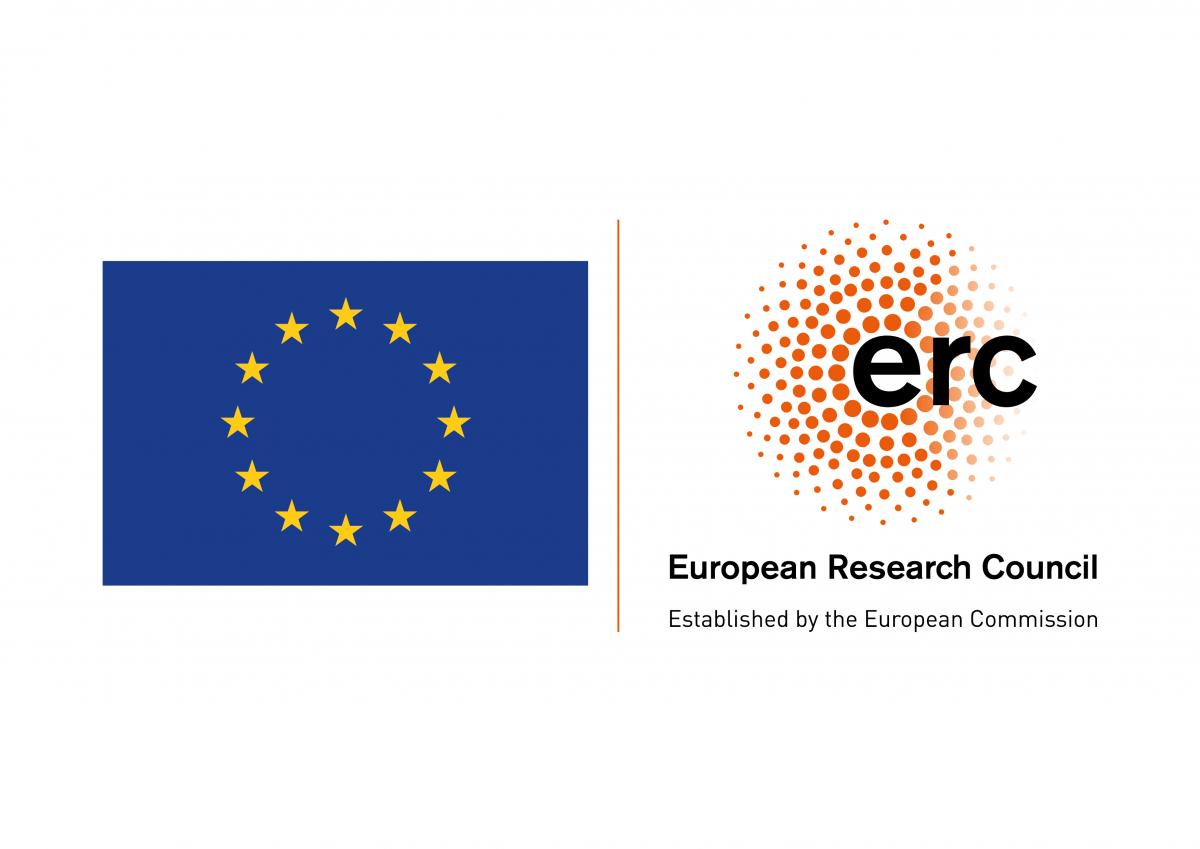}}\\\end{minipage}}}

\hypersetup{colorlinks=true,linkcolor=blue,citecolor=blue}

\makeatletter
\newcommand{\@abbrev}[3]{
  \def\c@a@def##1{
      \if ##1.
        \relax
      \else
        \@ifdefinable{\@nameuse{#1##1}}{\@namedef{#1##1}{#2##1}}
        \expandafter\c@a@def
      \fi
    }
  \c@a@def #3.
}
\@abbrev{bb}{\mathbb}{ABCDEFGHIJKLMNOPQRSTUVWXYZ}
\@abbrev{bf}{\mathbf}{ABCDEFGHIJKLMNOPQRSTUVWXYZabcdefghijklmnopqrstuvwxyz}
\@abbrev{bit}{\boldsymbol}{ABCDEFGHIJKLMNOPQRSTUVWXYZabcdefghijklmnopqrstuvwxyz}
\@abbrev{cal}{\mathcal}{ABCDEFGHIJKLMNOPQRSTUVWXYZ}
\@abbrev{frak}{\mathfrak}{ABCDEFGHIJKLMNOPQRSTUVWXYZabcdefghijklmnopqrstuvwxyz}
\@abbrev{rm}{\mathrm}{ABCDEFGHIJKLMNOPQRSTUVWXYZabcdefghijklmnopqrstuvwxyz}
\@abbrev{scr}{\mathscr}{ABCDEFGHIJKLMNOPQRSTUVWXYZ}
\@abbrev{sf}{\mathsf}{ABCDEFGHIJKLMNOPQRSTUVWXYZabcdefghijklmnopqrstuvwxyz}
\@abbrev{Alg}{\Alg}{ABCDEFGHIJKLMNOPQRSTUVWXYZ}
\@abbrev{Str}{\Str}{ABCDEFGHIJKLMNOPQRSTUVWXYZ}
\@abbrev{set}{\mathbb}{ABCDEFGHIJKLMNOPQRSTUVWXYZ}
\@abbrev{tup}{\tup}{ABCDEFGHIJKLMNOPQRSTUVWXYZabcdefghijklmnopqrstuvwxyz}
\@abbrev{tld}{\tilde}{ABCDEFGHIJKLMNOPQRSTUVWXYZabcdefghijklmnopqrstuvwxyz}
\makeatother

\usepackage{wasysym}
\usepackage{xspace}

\title{Rankwidth meets stability}

\author{
Jaroslav Ne\v set\v ril\\Charles University (IUUK), Praha, Czech Republic\\\small{\texttt{nesetril@iuuk.mff.cuni.cz}}
\and 
Patrice Ossona de Mendez\\CAMS (CNRS, UMR 8557), Paris, France\\\small{\texttt{pom@ehess.fr}}
\and 
Micha\l{} Pilipczuk\\University of Warsaw, Poland\\\small{\texttt{michal.pilipczuk@mimuw.edu.pl}}
\and 
Roman Rabinovich\\Technical University of Berlin, Germany\\\small{\texttt{roman.rabinovich@tu-berlin.de}}
\and
Sebastian Siebertz\\University of Bremen, Germany\\
\small{\texttt{siebertz@uni-bremen.de}}}

\date{}

\begin{document}
\maketitle

\begin{abstract}
\noindent
We study two notions of being well-structured for classes of graphs
that are inspired by classic model theory. A class of graphs $\Cc$ is
{\em{monadically stable}} if it is impossible to define arbitrarily
long linear orders in vertex-colored graphs from $\Cc$ using a fixed
first-order formula. Similarly, {\em{monadic dependence}} corresponds
to the impossibility of defining all graphs in this way.  Examples of
monadically stable graph classes are nowhere dense classes, which
provide a robust theory of sparsity. Examples of monadically dependent
classes are classes of bounded rankwidth (or equivalently, bounded
cliquewidth), which can be seen as a dense analog of classes of
bounded treewidth.  Thus, monadic stability and monadic dependence
extend classical structural notions for graphs by viewing them in a
wider, model-theoretical context. We explore this emerging theory by proving
the following:
\begin{itemize}
\item A class of graphs $\mathscr C$ is a first-order transduction of
  a class with bounded treewidth if and only if $\mathscr C$ has
  bounded rankwidth and a {\em{stable edge relation}} (i.e.\ graphs
  from $\mathscr C$ exclude some half-graph as a semi-induced
  subgraph).
\item If a class of graphs $\mathscr C$ is monadically dependent and
  not monadically stable, then $\mathscr C$ has in fact an unstable
  edge relation.
\end{itemize}
 
As a consequence, we show that classes with bounded rankwidth
excluding some half-graph as a semi-induced subgraph are linearly
$\chi$-bounded. Our proofs are effective and lead to polynomial time
algorithms.
\end{abstract}

\ERCagreement

\pagebreak
\section{Introduction}
The search for efficient algorithms has led to the study of the
structural properties of graph classes defined by the exclusion of
specific substructures.
%
For example, the structure theorem for graphs with excluded minors~\cite{RobertsonS03a} and for graphs with excluded topological minors
~\cite{grohe2015structure,Dvovrak2012} formed the basis of many
structural and algorithmic studies. A fundamental contribution of
these studies was to unveil the particular importance of classes with
bounded treewidth, which was confirmed by their specific algorithmic
properties. Precisely, Courcelle's theorem asserts that in classes
with bounded treewidth, every property definable in monadic
second-order logic~(MSO) can be tested
efficiently~\cite{courcelle1990monadic}.

Based on the exclusion of \emph{shallow} minors (or shallow
topological minors), two of the authors proposed a framework for the
structural study of classes of sparse graphs, namely {\em bounded
  expansion} classes and (more generally) {\em nowhere dense} classes
\cite{Sparsity}. This last notion of sparsity is characteristic to
monotone classes of graphs with fixed parameter tractable first-order
model checking \cite{dvovrak2013testing, grohe2017deciding}.

Much effort has been taken to extend the numerous algorithmic
applications of sparse graph classes, in particular of treewidth, to
dense graphs. For example, Courcelle's theorem was extended to classes
of bounded \emph{cliquewidth}~\cite{courcelle2000linear} (or
equivalently of bounded \emph{rankwidth} or bounded \emph{NLC-width}),
which is the dense analog of treewidth.

The move from sparse to dense is naturally followed by a move from
{\em monotone} classes (i.e.\ classes closed under subgraphs) to {\em
  hereditary} classes (i.e.\ classes closed under induced subgraphs).
Still, strong algorithmic properties are known to emerge when one
considers hereditary classes of graphs defined by forbidding simple
induced subgraphs (as witnessed by the class of cographs, circle
graphs, or perfect graphs), or semi-induced bipartite
subgraphs. Recall that a bipartite graph $H$ is a {\em semi-induced
  subgraph} of a graph $G$ if there exist two disjoint subsets of
vertices $A$ and $B$ of $G$ such that~$H$ is isomorphic to the
subgraph of $G$ with vertex set $A\cup B$ and all the edges present in
$G$ between~$A$ and $B$.

For example, the {\em VC-dimension} of a graph is defined from the
maximum size of a semi-induced subgraph isomorphic to a {\em powerset
  graph}, that is, to a bipartite graph with vertex set
$U\cup \mathcal P(U)$ and edge set
$\{xX~:~x\in U, X\in \mathcal P(U),\text{ and } x\in X\}$. Classes with bounded
VC-dimension are known to have specific statistical properties, which
are at the heart of computational learning
theory~\cite{angluin1992computational} and of numerous results in
algorithms in geometric graph theory (see
e.g.
\cite{Bronnimann1995,matousek2004bounded}).

A stronger assumption is that a graph excludes, as a semi-induced
subgraph, some {\em half-graph}: a~bipartite graphs with vertex set
$\{a_1,\dots,a_n\}\cup\{b_1,\dots,b_n\}$ and edge set
$\{a_ib_j~:~1\leq i\leq j\leq n\}$.  It has been observed that
half-graphs provide a primary example why irregular pairs cannot be
avoided in the statement of Szemer\'edi's Regularity Lemma. Indeed,
Malliaris and Shelah showed that forbidding a half-graph as a
semi-induced subgraph indeed makes it possible to get rid of irregular
pairs~\cite{malliaris2014regularity}.
\bigskip

In the language of model theory, a class excluding some powerset graph
as a semi-induced subgraph (that is, a class with bounded
VC-dimension) is said to have a {\em dependent edge relation}, and a
class excluding some half-graph as a semi-induced subgraph is said to
have a {\em stable edge relation} (or to have bounded \emph{order
  dimension}). This corresponds to the two main dividing lines used in
model theory: dependence and stability. In our setting, a class of
graphs is \emph{dependent} if every binary relation that is
(first-order) definable in it, seen as an edge relation, is
dependent. Similarly, a class is \emph{stable} if every definable
binary relation is stable. Stronger model theoretical notions are the
notions of {\em monadic dependence} and {\em monadic stability}, where
we restrict binary relations definable not only in graphs from the
class in question, but also in all their vertex-colorings.  A
surprising connection with structural graph theory is that, for a
monotone class of graphs, the properties of dependence, monadic
dependence, stability, monadic stability, and nowhere-denseness are
equivalent~\cite{adler2014interpreting}. However, without the
assumption of monotonicity, the notions of monadic dependence and
monadic stability do not collapse and present much wider concepts of
well-structuredness than nowhere denseness, and they are suited for
the treatment of dense graphs as well. For instance, every class of
bounded cliquewidth is monadically dependent~\cite{GroheT04}, but not
necessarily monadically
stable. 

One of our prime motivations is to extend the techniques designed for
classes of sparse graphs (i.e. bounded expansion or nowhere dense
classes) to the dense setting. For this, it is natural to consider
hereditary classes of graphs that are dependent, monadically
dependent, stable, or even monadically stable. As recently shown by
Fabia\'nski et al.~\cite{FabianskiPST19}, these structural assumptions
may be used in a novel way in the design of parameterized algorithms.

Monadic dependence and monadic stability can be also defined
using transductions. A~\emph{(first-order) transduction} is a way to
construct target graphs from vertex-colorings of source graphs by
fixed first-order formulas (see \cref{sec:prelims} for formal
definitions). In this setting, a class is monadically dependent if it
has no transduction onto the class of all powerset graphs
(equivalently, onto the class of all graphs). It is monadically
stable if it has no transduction onto the class of all
half-graphs~\cite{baldwin1985second}. From a dual point of view,
classes with bounded rankwidth are exactly those that are
transductions of the class of trivially perfect graphs (equivalently,
of the class of tree-orders). Similarly, classes with bounded linear
rankwidth are exactly those that are transductions of the
class of half-graphs (equivalently, of the class of linear
orders)~\cite{colcombet2007combinatorial}.

In this way, transductions form a basic containment notion for graphs,
which can be used to define structural properties through forbidding
obstructions, similarly to (shallow) minors or (induced)
subgraphs. The difference is that tranductions represent containment
understood in model-theoretical terms, and thus are suited for
considering questions related to first-order logic. As the notions of
monadic stability and monadic dependence are preserved by taking
transductions and they correspond to major dividing lines in model
theory, we expect them to be central in the emerging theory.

In order to explore this theory, it is imperative to understand
classical concepts of structural graph theory through the lense of
transductions. That is, we wish to describe the closures of
classes that are known to be well-structured under
transductions. This was done e.g. for classes of bounded
degree~\cite{GajarskyHOLR16} and for classes of bounded
expansion~\cite{SBE_TOCL}. More importantly for this work, in a
previous paper, the following characterization of monadically stable
classes of bounded linear rankwidth was~given.

\begin{theorem}[\cite{SODA_msrw}]
\label{thm:lrw}
If a class of graphs $\Cc$ has bounded linear rankwidth, then the
following conditions are equivalent:
\begin{enumerate}[nosep]
\item $\Cc$ has a stable edge relation;
\item $\Cc$ is stable;
\item $\Cc$ is monadically stable;
\item $\Cc$ is a transduction of a class with bounded pathwidth.
\end{enumerate}
\end{theorem}

Conceptually, this result means that if a class of graphs $\Cc$ has
bounded linear rankwidth and excludes some half-graph as a
semi-induced subgraph, then graphs from $\Cc$ can be ``sparsified'' in
the following sense: for each $G\in \Cc$ we can find a vertex-colored
graph $G'$ of bounded pathwidth such that~$G$ can be defined from $G'$
using fixed first-order formulas. The much more difficult question
whether a result analogous to \cref{thm:lrw} holds for classes of
bounded rankwidth (instead of {\em{linear}} rankwidth) could not be
answered in~\cite{SODA_msrw} and was stated there as a conjecture.

A by-product of the results of~\cite{SODA_msrw} is the conclusion that
classes of bounded linear rankwidth are linearly $\chi$-bounded.
Here, a hereditary class $\Cc$ of graphs is (linearly) $\chi$-bounded
if the chromatic number of graphs in $\Cc$ is functionally (linearly)
bounded by their clique number.  This concept was introduced by
Gy\'arf\'as~\cite{gyarfas1985problems} and has received a lot of
attention (see e.g.\ the surveys~\cite{Schiermeyer2019survey,
  scott2018survey}).

\paragraph*{Our contribution.}  In this work we prove the conjecture
stated in~\cite{SODA_msrw} and establish the following:

\begin{restatable}{theorem}{mainthm}
\label{thm:main}
If a class of graphs $\mathscr C$ has bounded rankwidth, then the
following conditions are equivalent:
\begin{enumerate}[nosep]
\item\label{e:sbt-er}  $\Cc$ has a stable edge relation;
\item\label{e:sbt-st} $\Cc$ is stable;
\item\label{e:sbt-mst} $\Cc$ is monadically stable;
\item\label{e:sbt-sbt} $\Cc$ is a transduction of a class with bounded
  treewidth.
\end{enumerate}
\end{restatable}

The implications
\enumref{e:sbt-sbt}$\Rightarrow$\enumref{e:sbt-mst}$\Rightarrow$\enumref{e:sbt-st}$\Rightarrow$\enumref{e:sbt-er}
are obvious.  For the implication
\enumref{e:sbt-er}$\Rightarrow$\enumref{e:sbt-sbt}, we combine the
approach presented in~\cite{SODA_msrw} with the techniques used by
Bonamy and the third author in~\cite{bonamy2019graphs} to prove that
classes of bounded rankwidth are polynomially $\chi$-bounded. Using
the tree variant of Simon's factorization due to
Colcombet~\cite{colcombet2007combinatorial}, the authors
of~\cite{bonamy2019graphs} introduce a bounded-depth recursive
decomposition of the tree encoding of a graph of rankwidth at most $k$
into factors, so that the quotient trees satisfy certain Ramsey
properties. We show that in the absence of a half-graph, these
properties imply that each root-to-leaf path in a quotient tree can be
partitioned into a bounded number of blocks, and the only ``points of
interest'' on the paths are borders between consecutive blocks.  This
leads to an encoding of the graph in question in a graph of bounded
treewidth, which can be decoded using fixed first-order formulas. We
stress that this encoding/decoding scheme is by no means
straightforward: it requires new combinatorial insights and a careful
analysis. The proof is constructive and can be implemented as a
polynomial time algorithm.
\medskip

Further, we show that the equivalence of the first three conditions of
\cref{thm:lrw} is in fact a more general phenomenon that occurs in
every monadically dependent graph class.  Precisely, we prove:

\begin{restatable}{theorem}{msthm}
\label{thm:ms}
For a monadically dependent graph class $\mathscr C$, the
following conditions are equivalent:
\begin{enumerate}[nosep]
\item\label{enum:hg} $\Cc$ has a stable edge relation;
\item\label{enum:s} $\Cc$ is stable;
\item\label{enum:ms} $\Cc$ is monadically stable.
\end{enumerate}	
\end{restatable}

 Note that implications
\enumref{enum:ms}$\Rightarrow$\enumref{enum:s}$\Rightarrow$\enumref{enum:hg}
are obvious. However, these implications can be strict for dependent
but not monadically dependent classes. For the implication
\enumref{enum:ms}$\Rightarrow$\enumref{enum:s} this is witnessed by
the class of \mbox{$1$-subdivided} half-graphs, which is dependent and
excludes some half-graph as a semi-induced subgraph, but is not
monadically stable.  For implication
\enumref{enum:s}$\Rightarrow$\enumref{enum:hg} this is witnessed by
the class of $1$-subdivided cliques, which is stable and thus
dependent, but is not monadically stable.

The proof of implication
\enumref{enum:hg}$\Rightarrow$\enumref{enum:ms} relies on the idea of
quantifier elimination.  Assuming that~$\Cc$ is not monadically
stable, we start with a formula $\varphi(\tup x,\tup y)$ that is
unstable in some {\em{monadic expansion}} $\Cc^+$ of~$\Cc$; that is,
$\Cc^+$ consists of graphs from $\Cc$ with some unary predicates
added.  Then we iteratively reduce $\varphi$ to simpler and simpler
unstable formulas while enriching~$\Cc^+$ with more unary predicates.
Eventually we find an atomic formula that is unstable on some monadic
expansion of $\Cc$, so $\Cc$ has an unstable edge
relation.  The assumption that $\Cc$ is monadically dependent is
crucially used in each quantifier elimination step.

\medskip

Moreover, \cref{thm:main} has important corollaries for classes with
\emph{low rankwidth covers/colorings} (introduced in~\cite{KwonPS20}).  It
follows from \cite{bonamy2019graphs} that classes with low rankwidth
covers are polynomially $\chi$-bounded.  Excluding a semi-induced
half-graph allows us to get a stronger property.
 
 \begin{restatable}{theorem}{thmrwcovchi}
   \label{thm:rwcovchi}
   Every class with low rankwidth covers and stable edge relation
   is linearly $\chi$-bounded.
 \end{restatable}
 In particular, \cref{thm:rwcovchi} implies that classes with
 bounded rankwidth and stable edge relation are linearly
 $\chi$-bounded.
 Also, requiring that a class has a stable edge relation gives
 the following~collapse. 
 
\begin{restatable}{theorem}{thmrwcov}
  \label{thm:rwcov}
  A class has low rankwidth covers and a stable edge relation if
  and only if it is a transduction of a class with bounded expansion.
 \end{restatable}
 \medskip
 
Our results together with observations present in the literature are
illustrated by the semi-lattice of properties of graph classes
 in \cref{fig:lattice}. See \cref{fig:lattice_ext}  in \cref{sec:conclusion} for an
extended version of the schema.

\begin{figure}[t]
\begin{minipage}{.5\textwidth}
\includegraphics[width=\textwidth]{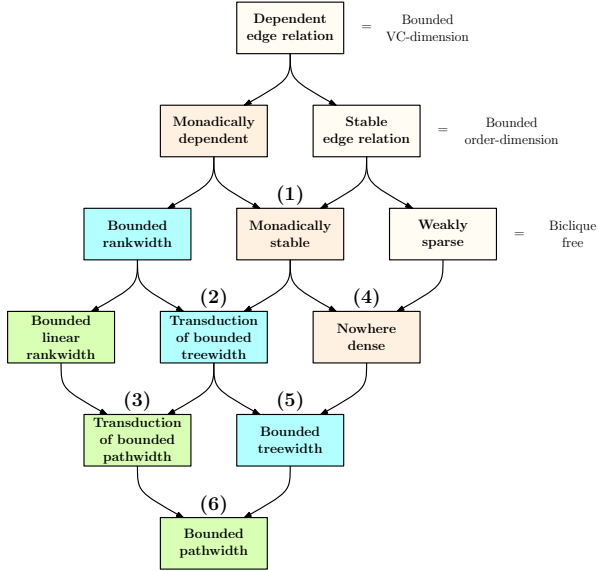}		
\end{minipage}
\begin{minipage}{.5\textwidth}
{\small
\begin{enumerate}[(1)]
	\item Monadically stable = monadically dependent $\cap$ stable edge relation (\cref{thm:ms});
	\item Structurally bounded treewidth = bounded rankwidth $\cap$ monadically stable (\cref{thm:main});
	\item Structurally bounded pathwidth = bounded linear rankwidth $\cap$ structurally bounded treewidth (follows from \cref{thm:lrw}, proved in \cite{SODA_msrw});
	\item Nowhere dense = monadically stable $\cap$ weakly sparse (follows from \cite{DVORAK2018143}, cf. \cite{msrw});
	\item Bounded treewidth = structurally bounded treewidth $\cap$ nowhere dense (= bounded rankwidth $\cap$ weakly sparse \cite{gurski2000tree});
	\item Bounded pathwidth = structurally bounded pathwidth $\cap$ bounded treewidth (= bounded linear rankwidth $\cap$ weakly sparse \cite{gurski2000tree}).
        \end{enumerate}
      }
\end{minipage}
\caption{The semi-lattice of property inclusions. }
\label{fig:lattice}
\end{figure}



\section{Preliminaries}\label{sec:prelims}

\paragraph*{Graphs.}
If $k$ is a positive integer, we write $[k]$ for the set
$\{1,\ldots,k\}$.
We consider finite, simple, undirected graphs.  For a graph $G$
we write~$V(G)$ for its vertex set and $E(G)$ for its edge set.

A graph $H$ is a {\em{subgraph}} of $G$ if
$V(H)\subseteq V(G)$ and $E(H)\subseteq E(G)$.  For $X\subseteq V(G)$,
we write $G[X]$ for the subgraph of $G$ induced by $X$, that is, the
subgraph with vertex set~$X$ and all edges from $G$ with both
endpoints in~$X$. A graph $H$ is an {\em{induced subgraph}} of $G$ if
there exists $X\subseteq V(G)$ such that $H$ is isomorphic to $G[X]$.
For disjoint subsets $X,Y$ of $V(G)$, we write $G[X,Y]$ for the
subgraph of~$G$ {\em semi-induced} by~$X$ and $Y$, that is, the
subgraph with vertex set $X\cup Y$ and all the edges of $G$ with one
endpoint in~$X$ and one endpoint in $Y$. A bipartite graph~$H$ is a
{\em semi-induced subgraph} of $G$ if $H$ is isomorphic to $G[X,Y]$
for some disjoint subsets~$X$ and~$Y$ of $V(G)$. A class $\Cc$ of
graphs excludes a bipartite graph $H$ as a semi-induced subgraph if no
$G\in \Cc$ contains $H$ as a semi-induced subgraph.

The complete bipartite graph
(biclique) with each side of size $t$ is denoted by $K_{t,t}$. The
\emph{half-graph} of order $t$ is the bipartite graph with vertices
$a_1,\ldots, a_t, b_1,\ldots, b_t$ and edges $a_ib_j$ for all
$i,j\in [t]$ with~$i\leq j$.

\paragraph*{First-order transductions.}  We assume familiarity with
first-order logic and refer to \cite{hodges1993model} for background.
We represent graphs as relational structures over a vocabulary
consisting of one binary edge relation symbol~$E$.  For a finite set
of unary relation symbols $\Sigma$, a {\em{$\Sigma$-expansion}} of a
graph~$G$ is a structure~$G^+$ obtained from $G$ by adding unary
relations with symbols in~$\Sigma$; thus, one can think of $G^+$ as of
$G$ with a coloring on the vertex set.  If we do not wish to specify
$\Sigma$, we may simply speak about a {\em{monadic expansion}} of $G$.
For a class $\Cc$ of graphs, a class~$\Cc^+$ is a {\em{monadic
    expansion}} of $\Cc$ if there is a finite set of unary relation
symbols $\Sigma$ such that every element of~$\Cc^+$ is a
$\Sigma$-expansion of a graph in $\Cc$.

For a formula~$\phi(\bar x)$ in the vocabulary of $\Sigma$-expanded
graphs, where $\bar x$ denotes a tuple of free variables, and a
$\Sigma$-expanded graph $G$, we define
$\phi(G)\coloneqq \{\,\bar u\in V(G)^{|\bar x|}\ \colon\ G\models\phi(\bar
u)\,\}$.  In particular, if~$A$ is a unary relation symbol, then
$A(G)=\{\,u\in V(G)\ \colon\ G\models A(u)\}$ and, as expected,
\mbox{$E(G)=\{\,(u,v)\in V(G)\times V(G)\ \colon\ G\models E(u,v)\,\}$}.
 
A {\em simple interpretation} $\mathsf I$ of graphs in
$\Sigma$-expanded graphs is a pair $(\nu(x), \eta(x,y))$ consisting of
two formulas (in the vocabulary of $\Sigma$-expanded graphs), where
$\eta$ is anti-reflexive and symmetric (i.e.~$\vdash \neg\eta(x,x)$
and $\vdash \eta(x,y)\leftrightarrow\eta(y,x)$). If $G^+$ is a
$\Sigma$-expanded graph, then $H=\mathsf I(G^+)$ is the graph with
vertex set~$\nu(G^+)$ and edge set
\mbox{$\eta(G^+)\cap(\nu(G^+)\times\nu(G^+))$}.
 
A \emph{transduction} $\Tr$ (from graphs to graphs) is a pair
$(\Sigma_{\Tr},\mathsf I_{\Tr})$, where $\Sigma_{\Tr}$ is a finite set
of unary relation symbols and $\mathsf I_{\Tr}$ is a simple
interpretation of graphs in \mbox{$\Sigma_{\Tr}$-expanded} graphs. A
graph $H$ can be \emph{$\Tr$-transduced} from a graph $G$ if there
exists a \mbox{$\Sigma_{\Tr}$-expansion}~$G^+$ of $G$ such that
$\mathsf I_{\Tr}(G^+)=H$. A class~$\Dd$ of graphs can be {\em
  $\Tr$-transduced} from a class $\Cc$ of graphs if for every graph
$H\in\Dd$ there exists a graph $G \in\Cc$ such that~$H$ can be
$\Tr$-transduced from $G$.  We also say that $\Tr$ is a
\emph{transduction from}~$\Cc$ \emph{onto} $\Dd$.  Note that if a
class $\Dd$ can be $\Tr$-transduced from a class $\Cc$ and
$\Dd'\subseteq \Dd$, then also $\Dd'$ can be $\Tr$-transduced from
$\Cc$. A class $\Dd$ of graphs can be {\em transduced} from a class
$\Cc$ of graphs if it can be $\Tr$-transduced from~$\Cc$ for some
transduction~$\Tr$. Note that transductions compose in the following
sense: If a class $\Dd$ can be transduced from a class $\Cc$ and a
class $\Ee$ can be transduced from $\Dd$, then $\Ee$ can be transduced
from $\Cc$.

\begin{remark}
  A class has bounded rankwidth if and only if it can be transduced
  from the class of trivially perfect graphs (i.e.\ from
  tree-orders)~\cite{colcombet2007combinatorial}. Hence, if a class
  $\Dd$ can be transduced from a class $\Cc$ of bounded rankwidth,
  then $\Dd$ has bounded rankwidth.
\end{remark}

\paragraph*{Stability and dependence.}  A formula
$\phi(\bar x,\bar y)$ is \emph{unstable} on a class $\Cc$ if for every
integer $n\geq 1$ there exists $G\in\Cc$,
$\bar a^1,\dots,\bar a^n\in V(G)^{|\bar x|}$ and
$\bar b^1,\dots,\bar b^n\in V(G)^{|\bar y|}$ such that
$G^+\models \phi(\bar a^i,\bar b^j)$ if and only if $i\leq j$.  The
formula $\phi(\bar x,\bar y)$ is \emph{stable} on $\Cc$ if it is not
unstable on $\Cc$. The class $\Cc$ has a \emph{stable edge relation}
if the formula $E(x,y)$ is stable on $\Cc$. The class $\Cc$ is
\emph{stable} if every formula $\phi(\bar x,\bar y)$ is stable on~
$\Cc$. The class~$\Cc$ is \emph{monadically stable} if every monadic
expansion $\Cc^+$ of $\Cc$ is stable.

Similarly, a formula $\phi(\bar x,\bar y)$ is \emph{independent} on a
class $\Cc$ if for every integer $n\geq 1$ there exists $G\in\Cc$,
$\bar a^1,\dots,\bar a^n\in V(G)^{|\bar x|}$ and
$\bar b^J\in V(G)^{|\bar y|}$ for all $J\subseteq [n]$ such that
$G^+\models \phi(\bar a^i,\bar b^J)$ if and only if~$i\in J$.  The
formula $\phi(\bar x,\bar y)$ is \emph{dependent} on $\Cc$ if it is
not independent on $\Cc$. The class $\Cc$ has a \emph{dependent edge
  relation} if the formula $E(x,y)$ is dependent on $\Cc$. The class
$\Cc$ is \emph{dependent} if every formula $\phi(\bar x,\bar y)$ is
dependent on~ $\Cc$. The class $\Cc$ is \emph{monadically dependent}
if every monadic expansion~$\Cc^+$ of $\Cc$ is dependent.

It turns out that monadic expansions allow us to circumvent the use of
tuples of variables $\bar x$ and $\bar y$ with length greater than
$1$, as stated next.

\begin{theorem}[follows from \cite{baldwin1985second}, see also
  \cite{anderson1990tree}]
\label{thm:stable_dep}
A class $\Cc$ is monadically dependent if and only if there is no
transduction from $\Cc$ onto the class of all finite graphs.
\end{theorem}

\begin{theorem}[\cite{baldwin1985second}]
\label{thm:stable_trans}
A class $\Cc$ is monadically stable if and only if there is no
transduction from $\Cc$ onto the class the all finite half-graphs.
\end{theorem}


\section{Rankwidth meets stability}
\label{sec:rw}
In this section we prove \cref{thm:main}. We start with some preliminaries on the toolbox introduced by Bonamy and the third author~\cite{bonamy2019graphs}, and then proceed to the proper proof.

\subsection{The toolbox}

\paragraph*{Trees.}  A \emph{tree} is a connected acyclic graph. A
\emph{rooted tree} is a tree $T$ with a distinguished node called the
\emph{root} of $T$, denoted $\top(T)$.  A rooted tree $T$ defines a
partial order on vertices and edges, which we denote by~$\preceq_T$ or
by $\preceq$ if $T$ is clear from the context. In this partial order
we have $\alpha\preceq\beta$ (with $\alpha,\beta\in V(T)\cup E(T)$) if
every path in $T$ that starts at the root and includes $\beta$ also
includes $\alpha$. If~$\alpha$ and~$\beta$ are nodes and
$\alpha\preceq \beta$, then we also say that $\alpha$ is an
{\em{ancestor}} of~$\beta$ and $\beta$ is a {\em{descendant}} of
$\alpha$; note that each node is considered also an ancestor of
itself. We also use terms {\em{parent}} and {\em{child}} with the
standard meaning. The parent of a node $v$ of a rooted tree (or the
node $v$ itself if $v$ is the root) is denoted by $v^{\uparrow}$; we
also denote $(v^\uparrow)^\uparrow$ by $v^{\uparrow\uparrow}$.  Note
that the ancestor partial order is an \emph{inf-semilattice}, with the
meet operation $\wedge$ being the least common ancestor.  The
\emph{leaves} of a rooted tree $T$ are the $\preceq$-maximal nodes of
$T$; the set of all leaves of $T$ is denoted by~$L(T)$.  Note that
from the perspective of first-order logic, a partial order is a
transitively oriented comparability graph. In particular, a tree-order
is a trivially perfect graph with a transitive orientation.

\paragraph*{$\mathcal{S}_k$-trees.} For a positive integer $k$ we let
$\mathcal{S}_k$ be the semigroup of all functions from~$[k]$ to~$[k]$
with composition as the semigroup operation.  That is, for
$f,g\in \mathcal{S}_k$ we write~$f\circ g\in \mathcal{S}_k$ for the
function that maps every $i\in [k]$ to $f(g(i))$. An element $f$ of a
semigroup is \emph{idempotent} if $f\circ f=f$.  An
\emph{$\mathcal{S}_k$-tree} is a tuple $(T,U,\rho,\pi)$, where $T$ is
a rooted tree, $U$ is a set, $\rho\colon E(T)\rightarrow \mathcal S_k$
is a labeling of the edges of $T$ by elements of $\mathcal{S}_k$, and
$\pi\colon U\rightarrow V(T)$ is a mapping from $U$ to the nodes of
$T$.

\paragraph*{Rankwidth, cliquewidth and NLC-width.}  There are various
equivalent ways of capturing the treelike structure of dense graphs
via hierarchical decompositions. The best known measures are probably
\emph{rankwidth}~\cite{oum2006approximating},
\emph{cliquewidth}~\cite{courcelle1993handle}, and
\emph{NLC-width}~\cite{wanke1994k}. All of these measures are
equivalent in the sense that if one measure is bounded on a class of
graphs, then the other measures are also
bounded~\cite{johansson1998clique,oum2006approximating}. We are going
to work with the following variant of NLC-width, which is easily seen
to be equivalent (in the above sense) to the original definition of
NLC-width.

Let $(T,U,\rho,\pi)$ be an $\mathcal{S}_k$-tree. For $x,y\in V(T)$
with $x\preceq_T y$ we denote by $\pth_T(y,x)=(e_1,\dots,e_s)$ the
sequence of edges on the unique path in $T$ from $y$ to $x$.  For
$v\in U$ and $x\preceq_T \pi(v)$ we further define
$\pth_T(v,x)\coloneqq \pth_T(\pi(v),x)$.  We implicitly extend $\rho$
to sequences of edges as follows:
$\rho((e_1,\dots,e_s))\coloneqq \rho(e_s)\circ\dots\circ\rho(e_1)$.

\begin{definition}
  Let $k$ be a positive integer and let $U$ be a set. A
  \emph{$k$-NLC-tree} on $U$ is a tuple \linebreak
  \mbox{$\Tf=(T,U,\rho,\pi, \eta, \chi)$}, where $(T,U,\rho,\pi)$ is
  an $\mathcal{S}_k$-tree,
  $\eta\colon V(T) \rightarrow 2^{[k]\times [k]}$ and
  $\chi\colon U\rightarrow[k]$.  We assume that $\eta$ is symmetric:
  for all $x\in V(T)$ and $(i,j)\in [k]\times [k]$, we have
  $(i,j)\in \eta(x)$ if and only if $(j,i)\in \eta(x)$.
\end{definition}

Let $\Tf=(T,U,\rho,\pi, \eta, \chi)$ be a $k$-NLC-tree.  We define the
\emph{color} in $T$ of $v\in U$ at a node $x\preceq_T\pi(v)$ of $T$ as
$\kappa_{\Tf}(v,x)\coloneqq\rho(\pth(v,x))(\chi(v))$.  The
$k$-NLC-tree $\Tf$ \emph{generates} the graph $G_{\Tf}$ with vertex
set~$U$, defined as follows: For $u\neq v\in U$, let
$x=\pi(u)\wedge_T \pi(v)$. Then $uv\in E(G)$ if and only if
$(\kappa_{\Tf}(u,x),\kappa_{\Tf}(v,x))\in \eta(x)$.

The \emph{NLC-width} of a graph $G$ is the minimum integer~$k$ such
that there exists a $k$-NLC-tree that generates~$G$ (see~\cref{fig:NLC} for an example of $k$-NLC-tree).

 \begin{figure}[ht]
 \begin{center}
   \includegraphics[width=.6\textwidth]{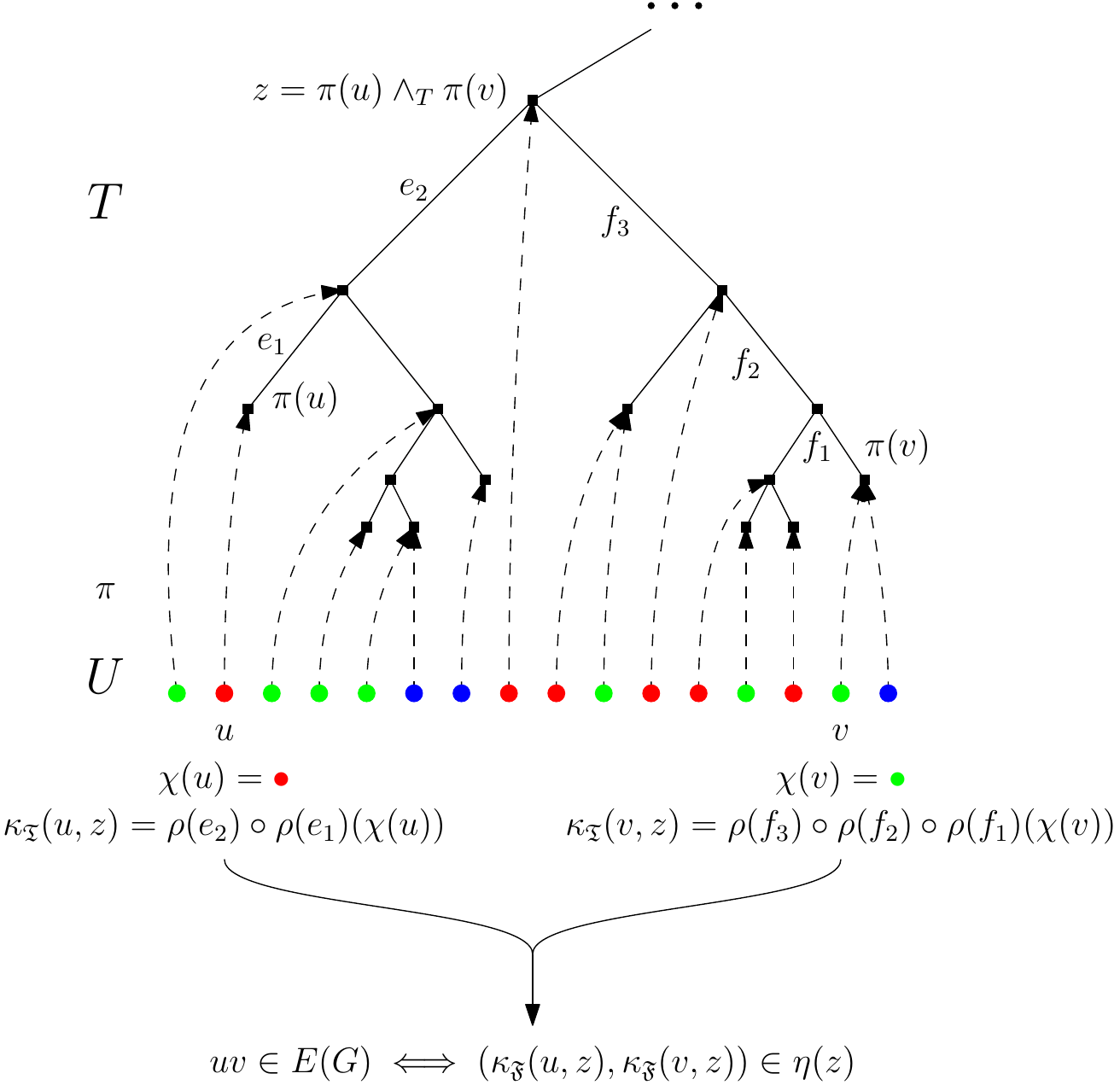}
 \end{center}
 \caption{A $k$-NLC-tree $\Tf=(T,U,\rho,\pi, \eta, \chi)$, and how the
   adjacency of two vertices $u$ and $v$ is determined.}
    \label{fig:NLC}
\end{figure}

Let $\Tf=(T,U,\rho,\pi, \eta, \chi)$ be a $k$-NLC-tree. Let $F$ be a
subtree of $T$ and let $\top(F)$ be the 
root of $F$, that is its $\preceq_T$-least
element. $F$ naturally induces a $k$-NLC-tree
$\Tf_F=(F,U_F,\rho_F,\pi_F,\eta_F,\chi_F)$, where
\mbox{$U_F\coloneqq \{u\in U\mid \pi(u)\succeq_T\top(F)\}$, $\rho_F$}
is the restriction of~$\rho$ to~$E(F)$, $\pi_F(v)$ (for $v\in U_F$) is
the \mbox{$\preceq_T$-maximum} element $x$ of $F$ with
$x\preceq_T \pi(v)$, $\eta_F$ is the restriction of~$\eta$ to~$V(F)$,
and $\chi_F(v)\coloneqq \kappa_{\Tf}(v,\pi_F(v))$.  Note that if $F'$
is a subtree of $F$, then $(\Tf_F)_{F'}=\Tf_{F'}$.

\begin{remark}
\label{rem:goto_factor}
Let $G_{\Tf}$ and $G_{\Tf_F}$ denote the graphs generated by $\Tf$ and
$\Tf_F$, respectively.  Then if for some $u,v\in U_F$ we have
$u\wedge_T v\in V(F)$, then $uv\in E(G_{\Tf})$ if and only if
$uv\in E(G_{\Tf_F})$.
\end{remark}

\begin{definition}
  A \emph{factorization} of\, $\Tf=(T,U,\rho,\pi, \eta, \chi)$ is a
  partition $\mathcal P$ of\, $T$ into vertex-disjoint subtrees.
\end{definition}
  
For a factorization $\mathcal P$ and a subtree $F\in \mathcal P$, the
$k$-NLC-tree $\Tf_F$ is called the \emph{factor of\, $\Tf$ induced
  by~$F$}. We define the \emph{quotient $\mathcal{S}_k$-tree}
$\Tf/\mathcal P=(Y,U,\varrho,\varpi)$ as follows (see \cref{fig:fact}):
\begin{itemize}[nosep]
\item $Y$ is the rooted tree with set of nodes $\mathcal P$, where $F$
  is an ancestor of $F'$ in $Y$ if and only if $\top(F)$ is an
  ancestor of $\top(F')$ in $T$ (i.e.
  $F\preceq_Y F'\iff \top(F')\preceq_T\top(F)$);
\item $\varrho$ is defined as
  $\varrho(F'F)=\rho(\pth_T(\top(F'),\top(F)))$, where $F$ is the
  parent of $F'$ in~$Y$;
\item $\varpi(v)$ is the tree $F\in\mathcal P$ that contains $\pi(v)$.
\end{itemize}  

\begin{figure}[ht]
\begin{center}
	\includegraphics[width=.7\textwidth]{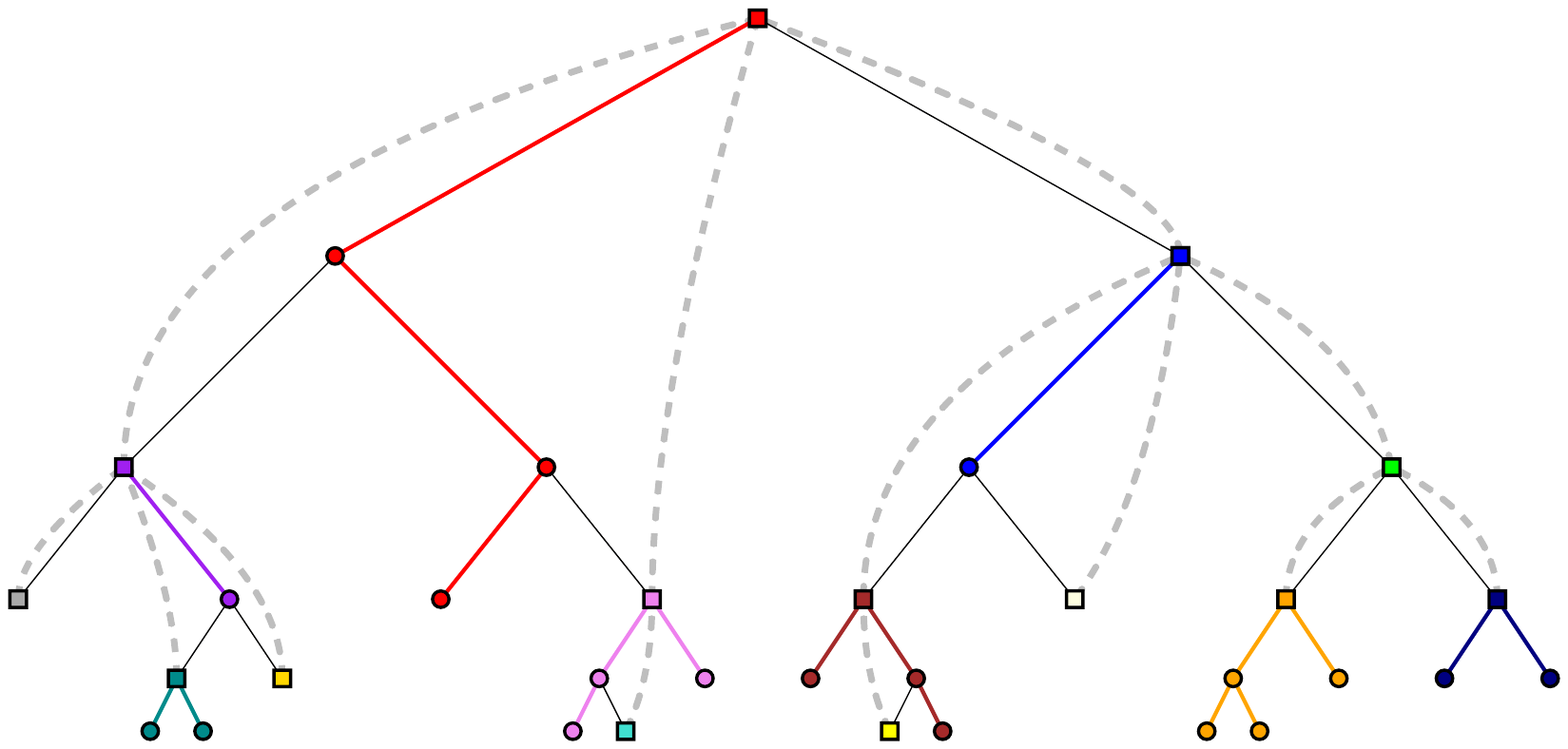}
\end{center}
\caption{Factors and quotient tree (dashed). The square nodes are the
  $\top$ nodes, and represent here the factors in the quotient tree.}
\label{fig:fact}
\end{figure} 

\begin{remark}
  Let $x\preceq_T \pi(v)$, and assume $x\in V(F)$, where
  $F\in\mathcal P$.  If $\pi(v)\in V(F)$, then we have
  $\rho(\pth_T(v,x))= \rho_F(\pth_F(v,x))$.  Otherwise, we have
\[		
  \rho(\pth_T(v,x))=\rho_F(\pth_F(v,x))\circ\rho(e(\top(F')))\circ\varrho(\pth_Y(\varpi(v),F'))\circ\rho_{\varpi(v)}(\pth_{\varpi(v)}(v,\top(\varpi(v)))),
\]
where $F'$ is the child of $F$ in $Y$ satisfying
$\top(F')\preceq_T\pi(v)$, and $e(\top(F'))$ is the edge that connects
$\top(F')$ with its parent in $T$.
\end{remark}

\paragraph*{Forward Ramsey and splendid trees.}  A set $A$ of elements
of~$\mathcal{S}_k$ is \emph{forward
  Ramsey}~\cite{colcombet2007combinatorial} if for all $e,f\in A$ we
have $e\circ f=e$.  In particular, each $e\in A$ is an idempotent in
$\mathcal{S}_k$, that is, $e\circ e=e$.  Note that if $A$ is forward
Ramsey, then it is a semigroup (as it is obviously closed by
composition).  An $\mathcal{S}_k$-tree $(T,U,\rho,\pi)$ is
\emph{splendid} if the set $\{\rho(e) \colon e\in E(T)\}$ is forward
Ramsey. It is \emph{shallow} if it has height $1$, i.e.\@ every
root-to-leaf path has at most one edge.
\medskip

The following lemma follows directly from \cite[Lemma 3.6]{bonamy2019graphs} (which is itself based on \cite{colcombet2007combinatorial}).
\begin{absolutelynopagebreak}
\begin{lemma}[{\cite[Lemma 3.6]{bonamy2019graphs}}]
\label{lem:factorization}
For every integer $k$ there exists a sequence of classes of
$k$-NLC-trees
$\mathcal{F}_0\subseteq \mathcal{F}_1\subseteq \ldots\subseteq
  \mathcal{F}_{3k^k}$ and a partition map
$\Tf\mapsto \mathcal P(\Tf)$, such that
\begin{enumerate}[(1),nosep]
\item $\mathcal{F}_0$ contains only single node $k$-NLC-trees, while
  $\mathcal{F}_{3k^k}$ is the class of all $k$-NLC-trees, and
\item for every $1\leq i\leq 3k^k$ and every $k$-NLC-tree
  $\Tf\in \mathcal F_i$, the factorization $\mathcal P(\Tf)$ of\,
  $\Tf$ is such that all the factors induced by parts of
  $\mathcal P(\Tf)$ belong to $\mathcal{F}_{i-1}$ and the quotient
  tree $\Tf/\mathcal{P}(\Tf)$ is either splendid or shallow.
\end{enumerate}
\end{lemma}
\end{absolutelynopagebreak}

Let $\Tf$ be a $k$-NLC-tree.  The map $\mathcal P(\cdot)$ defines a
\emph{recursive factorization} of $\Tf$, which can be represented as a
rooted tree, whose root is $\Tf$, where nodes are factors of $\Tf$,
and where the children of a factor $\Tf_F$ are the factors of $\Tf_F$
(thus of $\Tf$) induced by the parts of $\mathcal P(\Tf_F)$.

For a $k$-NLC-tree $\Tf$, the \emph{depth} of $\Tf$ is the minimum
integer $i$ such that $\Tf\in\mathcal F_i$.  Note that by
\cref{lem:factorization}, the depth of\ $\Tf$ is always upper bounded
by $3k^k$.


\subsection{Proof of Theorem~\ref{thm:main}}
\label{sec:mainthm}
In this section we prove that if a graph class $\Cc$ has bounded
rankwidth and stable edge relation, then~$\Cc$ can be transduced from
a class of bounded treewidth. Therefore, let us fix positive integers
$k$ and $h$ such that every graph in $\Cc$ admits a $k$-NLC-tree and
does not contain a half-graph of order $h$ as a semi-induced~subgraph.

We shall prove this inductively on the depth, as provided by
\cref{lem:factorization}.  More precisely, in the $i$th step of the
induction we prove that graphs from $\Cc$ that admit a $k$-NLC-tree
belonging to $\Ff_i$ can be transduced from a class of bounded
treewidth. Since the depth of any $k$-NLC-tree is bounded by $3k^k$,
the $3k^k$th step of the induction will end the proof of \cref{thm:main}.

Therefore, let us fix some graph $G$ and a $k$-NLC-tree
$\Tf=(T,U,\rho,\pi, \eta, \chi)$ generating $G$.  We let
$\Pp=\Pp(\Tf)$ be the factorization of $\Tf$ given by
\cref{lem:factorization}, and we denote by $(Y,U,\varrho,\varpi)$ the
quotient $\Ss_k$-tree~$\Tf/\Pp$.  Note that every factor of $\Pp$ has
depth lower than that of $\Tf$, hence we may apply the induction
assumption to it.

We first show how to handle the case when $(Y,U,\varrho,\varpi)$ is
splendid. Then we tackle the shallow case, which is significantly
simpler.  Each of these cases finishes with a technical claim
summarizing the analysis.  These claims are then used in a global
induction scheme.

\subsubsection{Splendid case}

As $(Y,U,\varrho,\varpi)$ is splendid, the set
$R=\{\varrho(e)\colon e\in E(Y)\}$ is forward Ramsey.  The following
lemma shows that the recolorings then have a particularly nice form.

\begin{lemma}[Claim~1 in Lemma~4.4 of~\cite{bonamy2019graphs}]
\label{lem:class}
Let $R\subseteq \Ss_k$ be forward Ramsey.  Then, for some $t\geq 1$,
$[k]$ can be partitioned into parts $\gamma_1,\ldots,\gamma_t$ so that
for every $f\in R$ and every $i\in [t]$ there exists $m_i\in \gamma_i$
such that $f(m)=m_i$ for all~$m\in \gamma_i$.
\end{lemma}

By applying \cref{lem:class} to $R$ we obtain a suitable partition
$\Gamma$ of $[k]$.  For $v\in U$, we let~$\gamma(v)$ be the part of
$\Gamma$ that contains $\kappa_{\Tf}(v,\top(\varpi(v)))$.  We call $v$
a {\em{$\gamma$-vertex}} if $\gamma(v)=\gamma$.

\paragraph*{Types and blocks.}
Throughout this section we use letters $x,y,z$ etc.\@ to denote the
nodes of~$Y$, which are parts of the factorization $\Pp$.  For a node
$x$ of $Y$ we denote by $P(x)$ the set of all ancestors of $x$ in~$Y$,
except for $x$ and its parent in $Y$, that is,
$P(x)=\{y\in V(Y)~:~y\preceq_Y x^{\uparrow\uparrow}\}$.  Recall that
nodes of~$Y$, being factors of $T$, are subtrees of $T$, hence it is
meaningful to say that a node $a$ of $T$ belongs to a node~$x$ of $Y$.

Let $x$ and $y$ be two nodes of $Y$ with $y\in P(x)$. Further, let
$\gamma\in \Gamma$.  Consider any vertex $v\in U$ satisfying
$x\wedge_Y \varpi(v) = y$, and let $a= \top(x)\wedge_T\pi(v)$. (Note
that $a$ is a vertex of $y$, considered as a subtree of $T$.)
Then we say that $x$ is {\em{$(\gamma,y)$-adjacent}} to $v$ 
if for some (equivalently, every) $m\in \gamma$ we have
\[\Bigl(\,\kappa_{\Tf}(v,a),\ \rho\bigl(\pth_T(\top(x),a)\bigr)(m)\,\Bigr)\in \eta(a)\,.\]
Otherwise, we shall say that~$x$ is {\em $(\gamma,y)$-non-adjacent} to
$v$.  Note here that by the properties of $\Gamma$ asserted by
\cref{lem:class}, the value of $\rho(\pth_T(\top(x),a))(m)$ does not
depend on the choice of $m\in \gamma$ whenever~$a$ does not belong to
$x$ or its parent in $Y$.

It may be useful to think of this definition as follows: if $u$ and
$v$ are vertices in $U$, $x=\varpi(u)$, $y=x\wedge_Y \varpi(v)$,
$y\in P(x)$, and $\kappa_{\Tf}(u,\top(x))\in \gamma$ (i.e.\ $u$ is a
$\gamma$-vertex), then $u$ is adjacent to $v$ in $G$ if and only if
$x$ is $(\gamma,y)$-adjacent to $v$.

\newcommand{\tp}{\mathsf{tp}}


Fix $\gamma_0,\gamma_1\in \Gamma$ and a node $x\in V(Y)$; possibly
$\gamma_0=\gamma_1$.  For every node $y\in P(x)$, we define the
{\em{$(\gamma_0,\gamma_1)$-type}} of~$y$ (seen from $x$), denoted
$\tp^x_{\gamma_0,\gamma_1}(y)$, as the pair $s_0s_1$ where $s_i$ is
set as the symbol \vspace{1pt}
\begin{itemize}[nosep]
\item $\void$ if there is no $\gamma_{1-i}$-vertex $v\in U$ satisfying
  $x\wedge_Y \varpi(v)=y$; and otherwise:
\item $+$ if $x$ is $(\gamma_i,y)$-adjacent to all the
  $\gamma_{1-i}$-vertices $v\in U$ satisfying $x\wedge_Y \varpi(v)=y$;
\item $-$ if $x$ is $(\gamma_i,y)$-non-adjacent to all the
  $\gamma_{1-i}$-vertices $v\in U$ satisfying $x\wedge_Y \varpi(v)=y$;
  and
\item $\pm$ otherwise.
\end{itemize}
The following lemma proves a basic synchronization property: for two nodes $x_0,x_1\in Y$, the types with respect to $x_0$ and $x_1$ synchronize above the parent of the least common ancestor of $x_0$ and $x_1$. 

\begin{Alemma}
\label{cl:sync_type}
If $x_0,x_1\in Y$ and $y\in P(x_0\wedge_Y x_1)$, then
$\tp^{x_0}_{\gamma_0,\gamma_1}(y)=\tp^{x_1}_{\gamma_0,\gamma_1}(y)$.
\end{Alemma}
\begin{proofEnd}
  Let $z=x_0\wedge_Y x_1$.  Consider any $m\in \gamma_0$. We have
 \begin{align*}
   \varrho(\pth_Y(x_0,z^\uparrow))(m)&=\varrho(\pth_Y(z,z^\uparrow))\circ\varrho(\pth_Y(x_0,z))(m)\\
                                     &=\varrho(\pth_Y(z,z^\uparrow))\circ\varrho(\pth_Y(x_1,z))(m)&\text{(by \cref{lem:class})}\\
                                     &=\varrho(\pth_Y(x_1,z^\uparrow))(m).
 \end{align*}
 Therefore, for every node $a\in V(T)$ that belongs to $y$ and is an
 ancestor of $\top(x_0)$ (equivalently, is an ancestor of
 $\top(x_1)$), we have
 \begin{align*}
   \rho(\pth_T(\top(x_0),a))(m) & = \rho(\top(z^\uparrow),a)\circ \varrho(\pth_Y(x_0,z^\uparrow))(m)\\
                                & = \rho(\top(z^\uparrow),a)\circ \varrho(\pth_Y(x_1,z^\uparrow))(m) \\
                                & = \rho(\pth_T(\top(x_1),a))(m).
 \end{align*}
 It follows that the first coordinates of
 $\tp^{x_0}_{\gamma_0,\gamma_1}(y)$ and
 $\tp^{x_1}_{\gamma_0,\gamma_1}(y)$ are equal.  That the second
 coordinates are equal as well follows from a symmetric reasoning.
\end{proofEnd}

The next lemma contains the key combinatorial observation of the
proof: a large alternation of types along $P(x)$ gives rise to a large
half-graph as a semi-induced subgraph.

\begin{lemma}\label{cl:small-alternation}
  Suppose in $P(x)$ there are nodes
  \[z_\ell\prec_Y y_\ell\prec_Y z_{\ell-1}\prec_Y y_{\ell-1}\prec_Y
    \ldots \prec_Y z_1\prec_Y y_1\] such that one of the following
  conditions holds:
  \begin{itemize}[nosep]
  \item for each $i\in [\ell]$, the first coordinate of
    $\tp^x_{\gamma_0,\gamma_1}(y_i)$ belongs to $\{+,\pm\}$ and the
    second coordinate of $\tp^x_{\gamma_0,\gamma_1}(z_i)$ belongs to
    $\{-,\pm\}$;
  \item for each $i\in [\ell]$, the first coordinate of
    $\tp^x_{\gamma_0,\gamma_1}(y_i)$ belongs to $\{-,\pm\}$ and the
    second coordinate of $\tp^x_{\gamma_0,\gamma_1}(z_i)$ belongs to
    $\{+,\pm\}$.
  \end{itemize}
  \vspace{1pt}
  Then $\ell\leq 3h$.
\end{lemma}
\begin{proof}
  Let us assume that the first of the two conditions holds, as the
  proof in the second case is analogous.  Suppose for contradiction
  that $\ell> 3h$.  By the definition of types, for each $i\in [\ell]$
  we can find a $\gamma_1$-vertex~$v_i$ that is
  $(\gamma_0,y_i)$-adjacent to $x$.  Similarly, for each $i\in [\ell]$
  we can find a $\gamma_0$-vertex $w_i$ that is
  $(\gamma_1,z_i)$-non-adjacent to $x$.  It easily follows from
  \cref{cl:sync_type} (see \cref{fig:hg1}) that vertices
 $$\{v_4,v_7,\ldots,v_{3h+1}\}\qquad\textrm{and}\qquad\{w_2,w_5,w_8,\ldots,w_{3h-1}\}$$
 semi-induce a half-graph of order $h$ in $G$, a contradiction.
\end{proof}

\begin{figure}[ht]
  \begin{center}
    \includegraphics[width=.9\textwidth]{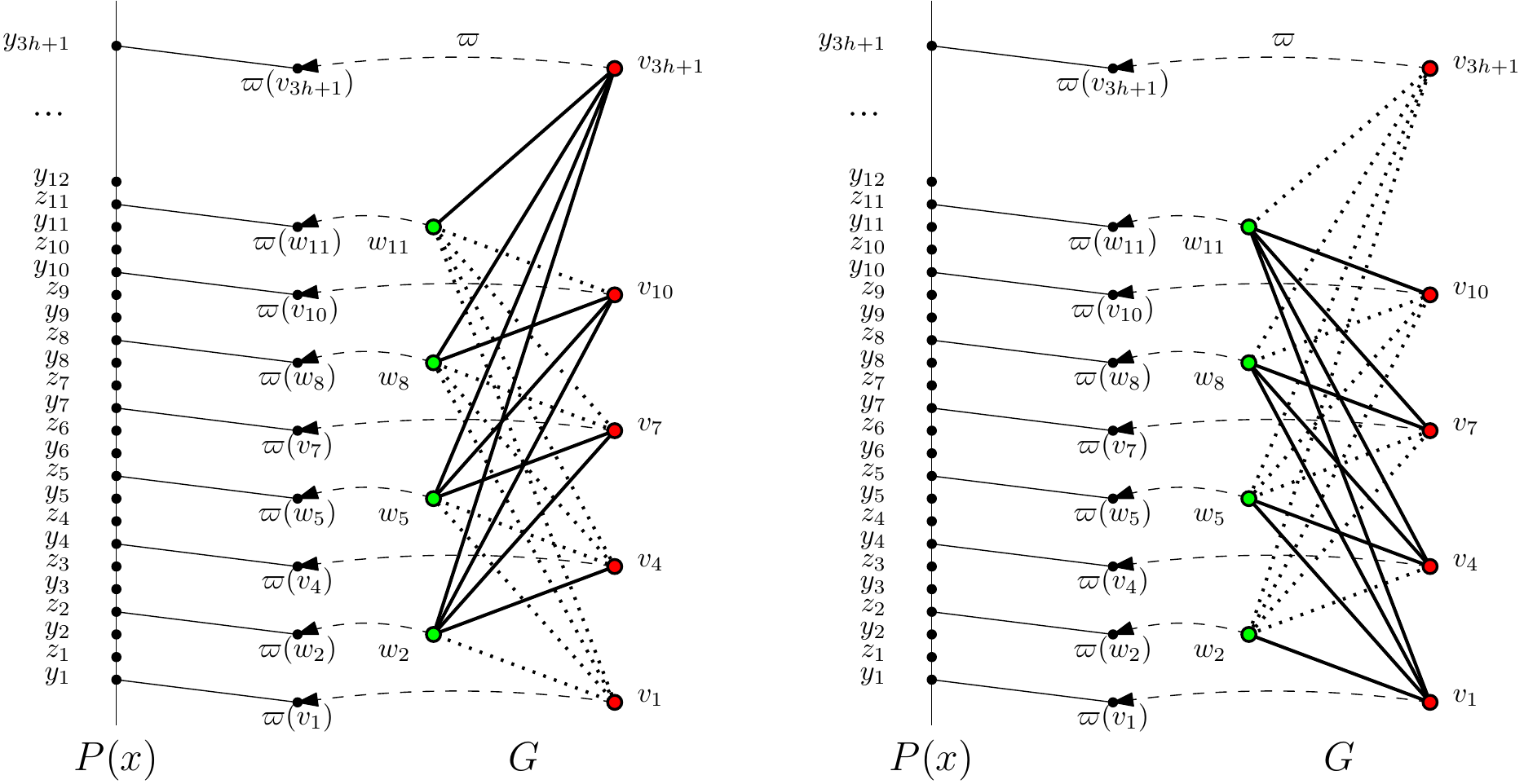}
  \end{center}
  \caption{Illustration for \cref{cl:small-alternation}.  On the left,
    the first coordinate of $\tp^x_{\gamma_0,\gamma_1}(y_i)$ belongs
    to $\{+,\pm\}$ and the second coordinate of
    $\tp^x_{\gamma_0,\gamma_1}(z_i)$ belongs to $\{-,\pm\}$; on the
    right, the first coordinate of $\tp^x_{\gamma_0,\gamma_1}(y_i)$
    belongs to $\{-,\pm\}$ and the second coordinate of
    $\tp^x_{\gamma_0,\gamma_1}(z_i)$ belongs to $\{+,\pm\}$.  The
    green vertices are of class $\gamma_0$, the red ones of class
    $\gamma_1$. Thin edges correspond to paths in $Y$; dashed arcs
    correspond to the $\varpi$ mapping; fat edges correspond to edges
    (plain) and non edges (dotted) of $G$.}
\label{fig:hg1}
\end{figure}

From \cref{cl:small-alternation} we may derive several structural
properties of the sequence of types of nodes on~$P(x)$.  We consider
$P(x)$ as a sequence ordered by the ancestor order, that is, the root
of $Y$ is the first element of this sequence.  Let then $\Types(x)$ be
the sequence of types $\tp^x_{\gamma_0,\gamma_1}(y)$ for $y\in P(x)$,
ordered as in $P(x)$.  In the following, by the {\em{type}} of
$y\in P(x)$ we mean the type $\tp^x_{\gamma_0,\gamma_1}(y)$.

We call an interval $J$ in the sequence $\Types(x)$ {\em{valid}} if it
is of one of the following kinds:
\begin{itemize}[nosep]
\item A {\em{fully mixed}} interval consists of a single node whose
  type does not contain $\void$, but either contains $\pm$ or both $+$
  and $-$.
\item A {\em{positive}} interval consists only of nodes with types in
  $\{\void\void,+\void,\void+,++\}$.
\item A {\em{negative}} interval consists only of nodes with types in
  $\{\void\void,-\void,\void-,--\}$.
\item A {\em{first-biased}} interval consists only of nodes with types
  in $\{\void\void,-\void,+\void,\pm\void\}$.
\item A {\em{second-biased}} interval consists only of nodes with
  types in $\{\void\void,\void-,\void+,\void\pm\}$.
\end{itemize}

Note that the cases are not exclusive.  An interval that is either
first- or second-biased will be just called {\em{biased}}.  Note that
a biased interval can be simultaneously positive and negative. If a
first-biased (resp.\ a second-biased) interval $J$ is neither positive
nor negative (that is, it includes a symbol $\pm$ or both symbols $+$
and $-$), then $J$ is called {\em{mixed-first-biased}} (resp.\
{\em{mixed-second-biased}}).

For a node $y\in P(x)$, let $J(y)$ be the longest valid interval in
$\Types(x)$ that starts at the position corresponding to the node $y$.
Then we define a partition $\Blocks(x)=\{A_1,A_2,\ldots\}$ of
$\Types(x)$ into subsequences, called {\em blocks}, via the following
greedy procedure: if $P(x)=y_1y_2\ldots y_\ell$, then
\begin{itemize}[nosep]
 \item $A_1=J(y_1)$;
 \item $A_2=J(y_{i_1+1})$, where $y_{i_1}$ is the last element of
   $A_1$;
 \item $A_3=J(y_{i_2+1})$, where $y_{i_2}$ is the last element of
   $A_2$, and so on.
\end{itemize}
The construction finishes once all the nodes of $P(x)$ are placed in
the blocks.  The blocks of $\Blocks(x)$ are naturally ordered as in
$\Types(x)$, i.e. $A_1$ contains $y_1$ that is the root of $Y$.

The following lemma shows that the number of blocks in the sequence
$\Blocks(x)$ is always bounded in terms of $h$ --- the order of the
half-graph that is forbidden in the graphs from $\Cc$. This is the key
observation of the proof and, up to a technical reasoning, it follows
from \cref{cl:small-alternation}: many blocks give rise to a large
half-graph in the generated graph.

\begin{Alemma}
\label{cl:num-blocks}
$\Blocks(x)$ contains at most $60h+9$ blocks.
\end{Alemma}
\begin{proofEnd}
  For contradiction suppose $\Blocks(x)$ contains more than $60h+9$
  blocks. Call a block {\em{mixed}} if it is either fully mixed, or
  mixed-first-biased, or mixed-second-biased.
 
 \begin{zclaim}
   $\Blocks(x)$ contains at most $12h+2$ fully-mixed blocks.
\end{zclaim}
\begin{claimproof}
  Suppose there are at least $12h+3$ fully mixed blocks in
  $\Blocks(x)$.  Recall that each fully mixed block consists of a
  single node of type belonging to
  $\{+-,-+,\pm-,\pm+,-\pm,+\pm,\pm\pm\}$.  Hence, we may either find
  at least $6k+2$ nodes with types in the set
  $\{+-,+\pm,\pm-,\pm\pm\}$, or at least $6h+2$ nodes with types in
  the set $\{-+,-\pm,\pm+,\pm\pm\}$.  In both cases, these at least
  $6h+2$ nodes form a structure that is forbidden by
  \cref{cl:small-alternation}, a contradiction.
\end{claimproof}

\begin{zclaim}
  $\Blocks(x)$ contains at most $6h+1$ mixed-first-biased blocks and
  at most $6h+1$ mixed-second-biased blocks.
\end{zclaim}
\begin{claimproof}
  We prove the bound on the number of mixed-first-biased blocks.  The
  bound for mixed-second-biased blocks follows analogously with the
  roles of $\gamma_0$ and $\gamma_1$ exchanged.

  Suppose for contradiction that there are more than $6h+1$
  mixed-first-biased blocks in $\Blocks(x)$.  Let
  $X_1,\ldots,X_{6h+2}$ be any $6h+2$ of them, ordered as in
  $\Types(x)$.  For $i\in [6h+1]$, let $z_i$ be the node of $P(x)$
  that immediately follows the last node of $X_i$.  Note that by the
  construction of $\Blocks(x)$, the second coordinate of the type of
  $z_i$ cannot be $\void$, for otherwise $z_i$ would be in $X_i$.  In
  particular, $z_i\notin X_{i+1}$ and $z_i$ lies in $P(x)$ strictly
  before~$X_{i+1}$.

  As argued, for each $i\in [6h+1]$ the second coordinate of the type
  of $z_i$ belongs to $\{-,+,\pm\}$.  Therefore, there exists a subset
  of indices $I\subseteq [6h+1]$ of size $3h+1$ such that either for
  each $i\in I$, the second coordinate of the type of $z_i$ belongs to
  $\{-,\pm\}$, or for each $i\in I$, the second coordinate of the type
  of $z_i$ belongs to $\{+,\pm\}$.  Assume the former case, as the
  proof in the latter case is symmetric.

  Since each $X_i$ is a mixed-first-biased block, for each $i\in I$ we
  may find a node $y_i\in X_i$ such that the first coordinate of the
  type of $y_i$ belongs to $\{+,\pm\}$. Now the nodes
  $\{y_i,z_i\colon i\in I\}$ form a structure forbidden by
  \cref{cl:small-alternation}, a contradiction.
\end{claimproof}

By the above claims, the total number of mixed blocks is at most
$24k+4$.  Call a block {\em{unaffected}} if it is not mixed and the
block succeeding it exists and is not mixed either.  Then the total
number of unaffected blocks is larger than
$(60h+9)-2\cdot (24h+4)-1=12h$.  Out of these, there are either more
than $6h$ unaffected positive blocks, or more than $6h$ unaffected
negative blocks.  Assume the former case, as the proof in the latter
case is symmetric.
 
Let then $B_1,\ldots,B_{6h+1}$ be any $6h+1$ unaffected positive
blocks, and let $C_1,\ldots,C_{6h+1}$ be the successors of blocks
$B_1,\ldots,B_{6h+1}$, respectively.  Since $B_1,\ldots,B_{6h+1}$ are
unaffected and positive, it follows that $C_1,\ldots,C_{6h+1}$ are
negative blocks.  Observe that for each $i\in [6h+1]$, it cannot
happen that for all the nodes $t\in B_i\cup C_i$, the first coordinate
of the type of $t$ is~$\void$. Indeed, then $B_i\cup C_i$ would be a
first-biased interval, and therefore it would be a valid interval that
would contain the block $B_i$ as a prefix.  Similarly, for each
$i\in [6h+1]$, it cannot happen that the second coordinate of the type
of $t$ is $\void$ for all $t\in B_i\cup C_i$. We conclude that for
each $i\in [6h+1]$, we may find nodes $y_i\in B_i$ and $z_i\in C_i$
such that one of the following alternatives holds:
 \begin{itemize}[nosep]
 \item the first coordinate of the type of $y_i$ is not $\void$ (and
   therefore must be $+$) and the second coordinate of the type of
   $z_i$ is not $\void$ (and therefore must be $-$); or
 \item the second coordinate of the type of $y_i$ is not $\void$ (and
   therefore must be $+$) and the first coordinate of the type of
   $z_i$ is not $\void$ (and therefore must be $-$).
 \end{itemize}
 By the pigeonhole principle, one of these two alternatives holds for
 at least $3h+1$ indices $i\in [6h+1]$.  Suppose this is the first
 alternative, as the proof in the other case proceeds analogously with
 the roles of $\gamma_0$ and $\gamma_1$ exchanged.  It now follows
 that if $I\subseteq [6h+1]$ is a set of size $3h+1$ such that the
 first alternative holds for each $i\in I$, then the nodes
 $\{y_i,z_i\colon i\in I\}$ form a structure forbidden by
 \cref{cl:small-alternation}, a contradiction.
\end{proofEnd}
\medskip

\noindent
For a node $x$ of $Y$, we define the following:
\begin{itemize}[nosep]
\item $Q(x)$ is the set consisting of $x$ and the parent of $x$ in
  $Y$, if existent;
\item $S(x)$ is the set containing, for each block $A\in \Blocks(x)$,
  the $\preceq_Y$-minimal element of~$A$, the $\preceq_Y$-minimal
  element of $A$ whose type belongs to $\{-,\pm\}$ (if existent), and
  the $\preceq_Y$-minimal element of $A$ whose type belongs to
  $\{+,\pm\}$ (if existent);
\item for each $\gamma\in \Gamma$, $g_\gamma(x)$ is the
  $\preceq_Y$-maximal ancestor of $x$ such that there exists a
  $\gamma$-vertex $w$ satisfying $g_\gamma(x)\preceq_Y \varpi(w)$, or
  $g_\gamma(x)=\bot$ if no such ancestor exists.
\end{itemize}
Further, let
$$\LL(x)=Q(x)\cup S(x).$$
By \cref{cl:num-blocks}, we have
$$|\LL(x)|\leq 2+3\cdot (60h+9)\leq 209h.$$ Intuitively,
$\LL(x)\cup \{g_{\gamma}(x)~:~\gamma\in \Gamma\}$ contains all vertices
that are interesting from the point of view of~$x$.

\paragraph*{Recovering edges: combinatorial analysis.}
Let us fix two vertices $u_0,u_1\in U$.  Let
\begin{eqnarray*}
  x_0=\varpi(u_0),& \qquad x_1=\varpi(u_1),\\
  \gamma_0=\gamma(u_0),&\qquad \gamma_1=\gamma(u_1). 
\end{eqnarray*}
Adopting the notation from the previous section, we have sets $P(x_0)$
and $P(x_1)$ and their partitions $\Blocks(x_0)$ and $\Blocks(x_1)$.
Intuitively, our goal is to show that given sets $\LL(x_0)$ and
$\LL(x_1)$, we may either directly infer whether $u_0$ and $u_1$ are
adjacent in $G$, or locate the node $z=x_0\wedge_Y x_1$, that is, the
lowest common ancestor of $x_0$ and $x_1$. In the subsequent section
we will implement this mechanism in first-order
logic. \cref{cl:sync_type} implies that the sequences of types
$\Types(x_0)$ and $\Types(x_1)$ agree on the prefix up to the
grandparent of $z$.

Let $Z$ be the set consisting of:
\begin{itemize}[nosep]
 \item $z$;
 \item the parent of $z$, if existent;
 \item the child of $z$ that is an ancestor of $x_0$, if existent; and
 \item the child of $z$ that is an ancestor of $x_1$, if existent.
\end{itemize}

We will further work under the following assumption:
\begin{equation}\label{p:lca-not-discovered}
  \LL(x_0)\cap Z=\emptyset\qquad\textrm{or}\qquad \LL(x_1)\cap Z=\emptyset. \tag{$\ast$}
\end{equation}
Intuitively, if assumption~\eqref{p:lca-not-discovered} is not
satisfied, then both $\LL(x_0)$ and $\LL(x_1)$ contain either $z$ or its
neighbor in~$Y$, and then locating $z$ will be easy.

Note that the root of $Y$ always belongs to $\LL(x_0)\cap \LL(x_1)$.
Hence, assuming~\eqref{p:lca-not-discovered}, $z$ is neither the root
of~$Y$ nor a child of the root of~$Y$.  Then both $P(x_0)$ and
$P(x_1)$ are non-empty, implying that also $\Blocks(x_0)$ and
$\Blocks(x_1)$ are non-empty.  Let
$$\Blocks(x_0)=\{A_1,A_2,\ldots,A_p\}\qquad\textrm{and}\qquad\Blocks(x_1)=\{B_1,B_2,\ldots,B_q\},$$
where blocks $A_i$ and $B_j$ are ordered naturally by the ancestor
order so that the root of $Y$ belongs to $A_1$ and $B_1$.  For a block
$A_i$, let $\top(A_i)$ be the first (i.e. $\preceq_Y$-minimal) node of
$A_i$; define $\top(B_j)$ analogously.

Let $i$ be the largest index such that $\top(A_i)=\top(B_i)$. Note
that $i$ is well-defined, because $\top(A_1)=\top(B_1)$.  Let
$t=\top(A_i)=\top(B_i)$. Since $t$ is both an ancestor of $x_0$ and of
$x_1$, we have $t\preceq_Y z$.  Furthermore, since
$t\in \LL(x_0)\cap \LL(x_1)$, from~\eqref{p:lca-not-discovered} we infer
that $t\notin Z$.
\begin{Alemma}
\label{lem:z}
The node $z$ has the following properties:
\begin{enumerate}
\item\label{enum:tpz1} $z\in Q(x_0)$ or the first coordinate of
  $\tp^{x_0}_{\gamma_0,\gamma_1}(z)$ is not equal to $\void$;
\item\label{enum:tpz2} $z\in Q(x_1)$ or the second coordinate of
  $\tp^{x_1}_{\gamma_0,\gamma_1}(z)$ is not equal to $\void$;
\item\label{enum:covered} $z\in A_i\cup B_i$.
\end{enumerate}
\end{Alemma}
\begin{proofEnd}
 The first two points follow directly from the existence of vertices $u_0$ and $u_1$. We are left with arguing that $z\in A_i\cup B_i$.
%
  Suppose otherwise. Then both $A_{i+1}$ and $B_{i+1}$ exist, and
  moreover $\top(A_{i+1})\prec_Y z$ and $\top(B_{i+1})\prec_Y z$.  By
  the maximality of $i$ we have $\top(A_{i+1})\neq \top(B_{i+1})$.

  By \cref{cl:sync_type} and the construction of $\Blocks(x_1)$ and
  $\Blocks(x_2)$, every ancestor of the grandparent of $z$ is the top
  vertex of a block in $\Blocks(x_1)$ if and only if it is the top
  vertex of a block in $\Blocks(x_2)$.  Therefore,
  $\top(A_{i+1})\neq \top(B_{i+1})$ together with
  $\top(A_{i+1})\prec_Y z$ and $\top(B_{i+1})\prec_Y z$ implies that
  $\top(A_{i+1})\in Z$ and $\top(B_{i+1})\in Z$. As
  $\top(A_{i+1})\in \LL(x_0)$ and $\top(B_{i+1})\in \LL(x_1)$, this
  contradicts assumption~\eqref{p:lca-not-discovered}.
\end{proofEnd}

\noindent
Let
$R\coloneqq \{\,r~:~t\preceq_Y r \preceq_Y z\textrm{ and } r\notin
Z\,\}$.  Note that $t\in R$, hence $R$ is
non-empty. By~\cref{cl:sync_type}, we have
\begin{equation}\label{eq:sameR}
  \tp^{x_0}_{\gamma_0,\gamma_1}(r)=\tp^{x_1}_{\gamma_0,\gamma_1}(r)\qquad\textrm{for each }r\in R. 
\end{equation}
From the construction of $\Blocks(x_0)$ and of $\Blocks(x_1)$ it then
follows that
\begin{equation}\label{eq:Rcontained}
R\subseteq A_i\cap B_i. 
\end{equation}
We now observe the following.
\begin{Alemma}
\label{cl:R-non-void}
 There exists $r\in R$ such that
 $$\tp^{x_0}_{\gamma_0,\gamma_1}(r)=\tp^{x_1}_{\gamma_0,\gamma_1}(r)\neq \void\void.$$
\end{Alemma}
\begin{proofEnd}
  Suppose otherwise:
  $\tp^{x_0}_{\gamma_0,\gamma_1}(r)=\tp^{x_1}_{\gamma_0,\gamma_1}(r)=\void\void$
  for all $r\in R$.  By 
  \cref{lem:z}, we either have $z\in Q(x_0)$, or
  $\tp^{x_0}_{\gamma_0,\gamma_1}(z)\neq \void\void$.  The latter
  condition implies that either $A_{i+1}$ exists and
  $\top(A_{i+1})\in Z$, or the $\preceq_Y$-minimal element of block
  $A_i$ whose type features a non-$\void$ symbol belongs to $Z$.  In
  each of these three cases we have $\LL(x_0)\cap Z\neq \emptyset$. A
  symmetric reasoning shows that also $\LL(x_1)\cap Z\neq\emptyset$.
  This is a contradiction with
  assumption~\eqref{p:lca-not-discovered}.
\end{proofEnd}

\noindent

We introduce the following notation.  For $\gamma\in \Gamma$ and
$y\in V(Y)$, if there is a unique grandchild~$y'$ of $y$ in $Y$ such
that for every $\gamma$-vertex $v$ satisfying $y\preceq_Y \varpi(v)$
we have $y'\preceq_Y \varpi(v)$, then we set $h_\gamma(y)=y'$.  If
there is no such grandchild, we set $h_\gamma(y)=\bot$.

\begin{Alemma}
\label{lem:thecases}
None of the blocks $A_i$ or $B_i$ is fully mixed. Moreover, depending
on the kinds the blocks~$A_i$ and $B_i$ belong to, we have the
following cases:
\begin{enumerate}
\item\label{enum:A_not_biased} If $A_i$ is not biased,
  then \begin{itemize}
  \item either $A_i$ is positive and $u_0u_1\in E(G)$,
  \item or $A_i$ is negative and $u_0u_1\notin E(G)$.
  \end{itemize}
\item\label{enum:B_not_biased} If $B_i$ is not biased,
  then \begin{itemize}
  \item either $B_i$ is positive and $u_0u_1\in E(G)$,
  \item or $B_i$ is negative and $u_0u_1\notin E(G)$.
  \end{itemize}
\item\label{enum:both_biased} If both $A_i$ and $B_i$ are biased, then
  \begin{itemize}
  \item either both $A_i$ and $B_i$ are first-biased, and then
    $h_{\gamma_0}(z)\neq\bot$,
  \item or both $A_i$ and $B_i$ are second-biased and then
    $h_{\gamma_1}(z)\neq\bot$.
  \end{itemize}
\end{enumerate}
\end{Alemma}
\begin{proofEnd}
 First, we observe the following.

\begin{zclaim}
 None of the blocks $A_i$ or $B_i$ is fully mixed.
\end{zclaim}
\begin{claimproof}
  Recall that a fully mixed block consists of one node whose type does
  not feature symbol~$\void$, but features either $\pm$ or both $+$
  and $-$.  Therefore, if any of $A_i$ or $B_i$ was fully mixed, then
  both of them would be, implying that $A_i=B_i=\{t\}$.  This stands
  in contradiction with \cref{lem:z}. 
\end{claimproof}

Next, we treat the case when  $A_i$ or $B_i$ is not biased.

\begin{zclaim}
  Suppose $A_i$ is not biased. Then exactly one of the following
  holds: $A_i$ is positive and $u_0u_1\in E(G)$, or $A_i$ is negative
  and $u_0u_1\notin E(G)$.  Symmetrically, supposing $B_i$ is not
  biased, exactly one of the following holds: $B_i$ is positive and
  $u_0u_1\in E(G)$, or $B_i$ is negative and $u_0u_1\notin E(G)$.
\end{zclaim}
\begin{claimproof}
  We prove the first assertion; the reasoning proving the second one
  is symmetric.
 
  By the previous claim and the assumption, $A_i$ is neither fully
  mixed, nor first-biased, nor second-biased.  Therefore, $A_i$ is
  either positive or negative. Note that by \cref{cl:R-non-void}
  and~\eqref{eq:Rcontained}, $A_i$ cannot be both positive and
  negative at the same time.  It remains to prove that if $A_i$ is
  positive, then $u_0u_1\in E(G)$; the proof that $A_i$ being negative
  entails $u_0u_1\notin E(G)$ is symmetric.
 
  Note that if we have $z\in A_i$, then $A_i$ being positive
  immediately implies that $u_0u_1\in E(G)$.  Therefore, suppose that
  $z\notin A_i$, which implies that
$z\preceq \top(A_{i+1})$ and as
    $\top(B_{i+1}) \neq \top(A_{i+1})$ (by definition of $i$),
    $z\in \LL(x_0)$ and thus $\LL(x_0)\cap Z\neq \emptyset$.  By
  \cref{lem:z}, 
  we have $z\in B_i$.  Suppose for contradiction that
  $u_0u_1\notin E(G)$.  Then the second coordinate of
  $\tp^{x_1}_{\gamma_0,\gamma_1}(z)$ has to be either $-$ or $\pm$.
  However, since~$A_i$ is positive, from~\eqref{eq:sameR}
  and~\eqref{eq:Rcontained} we infer that types
  $\tp^{x_1}_{\gamma_0,\gamma_1}(r)$ for $r\in R$ feature only symbols
  $\void$ and $+$. Therefore, the $\preceq_Y$-minimal element of $B_i$
  that contains symbol $-$ or $\pm$ is either $z$ or its parent,
  implying that $\LL(x_1)\cap Z\neq \emptyset$. Together with
  $\LL(x_0)\cap Z\neq \emptyset$, this contradicts
  assumption~\eqref{p:lca-not-discovered}.
\end{claimproof}

We are left with the case when both $A_i$ and $B_i$ are biased. First,
we observe that they need to be biased in the same direction.

\begin{zclaim}
  If both $A_i$ and $B_i$ are biased, then exactly one of the
  following holds: both $A_i$ and $B_i$ are first-biased, or both
  $A_i$ and $B_i$ are second-biased.
\end{zclaim}
\begin{claimproof}
  Follows directly from \cref{cl:R-non-void} together
  with~\eqref{eq:Rcontained}.
\end{claimproof}

We now show how to locate $z$ in this case.

\begin{zclaim}
  Suppose $A_i$ and $B_i$ are both first-biased.  Then $z\in A_i$ and
  $z=g_{\gamma_0}(x_1)$; in particular $z\in P(x_0)$.  Moreover, there
  exists a grandchild $z_0$ of $z$ such that for every
  $\gamma_0$-vertex $v$ satisfying $z\preceq_Y \varpi(v)$, we in fact
  have $z_0\preceq_Y\varpi(v)$.  Also, there exist $\gamma_0$-vertices
  satisfying this condition.
 
 In other words, $h_{\gamma_0}(z)=z_0$.
\end{zclaim}
\begin{claimproof}
  Since $B_i$ is first-biased, from~\cref{lem:z}(2) we infer that
  $z\in Q(x_1)$. It implies that $\LL(x_1)\cap Z \neq \emptyset$ and
  that $z\notin P(x_1)$ thus $z\notin B_i$. By \cref{lem:z}(3),
  $z\in A_i$.
  As $z\notin B_i$ and $R\subseteq B_i$, we have
  $\LL(x_1)\cap Z\neq \emptyset$.  Therefore, from
  assumption~\eqref{p:lca-not-discovered} we conclude that
  $\LL(x_0)\cap Z=\emptyset$.
 
  As $z\in A_i$, we in particular have $z\in P(x_0)$, hence $z$ is
  neither $x_0$ nor the parent of $x_0$.  Let then $z_0$ be the
  grandchild of $z$ such that $z_0\preceq_Y x_0$.  Further, let $z'_0$
  be the parent of $z_0$.  Note that $z'_0\in Z$. Since
  $\LL(x_0)\cap Z=\emptyset$, we must have $z'_0\in A_i$.
 
  Since $A_i$ is first-biased and $z,z'_0\in A_i$, the second
  coordinates of $\tp^{x_0}_{\gamma_0,\gamma_1}(z)$ and of
  $\tp^{x_0}_{\gamma_0,\gamma_1}(z_0')$ are both~$\void$.  Therefore,
  there are no $\gamma_0$-vertices $v$ satisfying
  $x_0\wedge_Y \varpi(v)=z$ or $x_0\wedge_Y \varpi(v)=z_0'$, which
  means that for every $\gamma_0$-vertex $v$ satisfying
  $z\preceq_Y \varpi(v)$, we in fact have $z_0\preceq_Y\varpi(v)$.
  That there exist $\gamma_0$-vertices satisfying this condition is
  witnessed by~$u_0$.
\end{claimproof}

A symmetric reasoning yields the following.

\begin{zclaim}
  Suppose $A_i$ and $B_i$ are both second-biased. Then $z\in B_i$ and
  $z=g_{\gamma_1}(x_0)$; in particular $z\in P(x_1)$.  Moreover, there
  exists a grandchild $z_1$ of $z$ such that for every
  $\gamma_1$-vertex $v$ satisfying $z\preceq_Y \varpi(v)$, we in fact
  have $z_1\preceq_Y\varpi(v)$.  Also, there exist $\gamma_1$-vertices
  satisfying this condition.
 
 In other words, $h_{\gamma_1}(z)=z_1$.
\end{zclaim}

The presented claims verify all the assertions from the lemma statement.
\end{proofEnd}

\paragraph*{Recovering edges: logical implementation.}
We now define a structure $H_{\Tf}$ which encodes all the relevant
information about the $k$-NLC-tree $\Tf$ and its factorization $\Pp$.
Intuitively, $H_{\Tf}$ encodes $\Tf$ in the natural way, plus in
addition we enrich it with pointers encoding sets $\LL(x)$ and
functions $g_\gamma(x),h_\gamma(x)$.

Formally, the universe of $H_{\Tf}$ is just $V(T)$; note that the set
$U$ will {\em{not}} be directly encoded.  In $H_{\Tf}$ we will use
only unary predicates and unary (partial) functions. Of course, the
latter can be replaced by suitable functional binary relations in
order to make the signature purely relational.  In the following,
whenever we encode some node $y$ that belongs to the quotient tree
$Y$, we represent it using $\top(y)$.  For instance, the parent
function in $Y$ is represented as a partial function on the nodes of
$T$ that maps $\top(x)$ to $\top(x')$ whenever $x'$ is the parent of
$x$ in $Y$.

\newcommand{\wL}{\widehat{\LL}}
\newcommand{\wg}{\widehat{g}}

For $x\in V(Y)$, let $\wL(x)\subseteq V(Y)$ be the set containing
every ancestor of $x$ that:
\begin{itemize}[nosep]
 \item belongs to $\LL(x)$,
 \item is the parent of a node of $\LL(x)$,
 \item is the child of a node of $\LL(x)$ on $P(x)$, or
 \item is the grandchild of a node of $\LL(x)$ on $P(x)$.
\end{itemize}
Recalling that $|\LL(x)|\leq 209h$, we have $|\wL(x)|\leq 836h$.  Also,
for $x\in V(Y)$ and $\gamma\in \Gamma$, we let $\wg_\gamma(x)$ be the
child of $g_\gamma(x)$ that is an ancestor of $x$.  In case
$g_\gamma(x)=x$, we set $\wg_\gamma(x)=\bot$.

In the following encoding, all values featuring $\bot$ are removed
from the domains of corresponding mappings.  Then, in $H_{\Tf}$ we
encode:
\begin{itemize}[nosep]
\item the parent function of the tree $T$;
\item the parent function of the tree $Y$;
\item the mapping $a\mapsto \rho(e(a))$, where $a$ is a node of $T$
  and $e(a)$ is the edge of $T$ connecting $a$ with its parent;
\item the mapping $x\mapsto \varrho(e(x))$, where $x$ is a node of $Y$
  and $e(x)$ is the edge of $Y$ connecting $x$ with its parent;
\item the mappings $a\mapsto \top(x(a))$ and
  $a\mapsto \rho(\pth_T(a,\top(x(a))))$, where $a$ is a node of $T$
  and $x(a)$ is the node of $Y$ such that $a\in x(a)$;
 \item for each $\gamma\in \Gamma$, the mappings $x\mapsto g_\gamma(x)$, $x\mapsto \wg_\gamma(x)$, and $x\mapsto h_\gamma(x)$;
 \item for each $\gamma\in \Gamma$, the mapping
   $x\mapsto \varrho(\pth_Y(x,\wg_{\gamma}(x)))$;
 \item the mapping $x\mapsto \wL(x)$, together with relevant data about the elements of $\wL(x)$; and
 \item for every node $x$ of $Y$ and $y\in \wL(x)$, the value
   $\varrho(\pth_Y(x,y))$.
\end{itemize}
Here, the last two points require more explanation.  Recall that
$|\wL(x)|\leq 836h$ for each $x\in V(Y)$. Therefore, to encode the
mapping $x\mapsto \wL(x)$ we use $836h$ distinct unary functions,
where the $i$th function maps a node $x\in V(Y)$ to the $i$th element
of $\wL(x)$, sorted by the ancestor order.  The relevant data about a
node $y\in \wL(x)$ includes whether $y$ is the $\preceq_Y$-minimal
node of some block of $\Blocks(x)$ and if so, what kind of block it is
(positive or negative, first-biased or second-biased, etc.). This
information can be encoded using unary predicates at $x$.  Similarly,
to encode the values $\varrho(\pth_Y(x,y))$ for $y\in \wL(x)$, we use
$836h$ distinct unary predicates at $x$, where the $i$th predicate
encodes $\varrho(\pth_Y(x,y))$ where $y$ is the $i$th element of
$\wL(x)$.

We later use some properties of $H_\Tf$ that follow from the
synchronization property expressed by \cref{cl:sync_type}. For this,
for a node $a$ of $T$, we define $N^\uparrow(a)$ to be the set of all
nodes $b$ of $T$ such that $b\prec_T a$ and there is a function $f$ in
$H_\Tf$ such that $b=f(a)$ or $a=f(b)$. Then we have the following.

\begin{Alemma}	
\label{cl:SReach-splendid}
 For each $a\in V(T)$, 
 $$\Bigl|\{b\in V(T)\colon b\prec_T a\}\cap \bigcup_{a' \succeq_T a} N^\uparrow(a')\Bigr|\leq 836h+2k+4.$$	
\end{Alemma}
\begin{proofEnd}
  Let $x\in V(Y)$ be such that $a\in x$.  From \cref{cl:sync_type} and
  the construction of the blocks it follows that for all $x',x''\in V(Y)$ such that $\top(x'),\top(x'')\succeq_T a$, we have
  $$\LL(x')\cap P(x)= \LL(x'')\cap P(x).$$ Thus,
  $$\wL(x')\cap P(x^\uparrow)= \wL(x'')\cap P(x^\uparrow)$$ for all such $x',x''$.
  Let $M_0$ be this common subset of~$P(x^\uparrow)$; note that $\wL(x')\cap P(x)\subseteq M_0\cup\{x^{\uparrow\uparrow}\}$. Let $M_0'=\{\top(z):z\in M_0\}\cup\{\top(x^{\uparrow\uparrow})\}$; then $|M_0'|\leq 836h+1$.
  
 Similarly, for all $x',x''$ as above, we have
 $$\{g_\gamma(x'),\wg_\gamma(x')\colon \gamma\in \Gamma\}\cap P(x)=\{g_\gamma(x''),\wg_\gamma(x'')\colon \gamma\in \Gamma\}\cap P(x),$$
 so let $M_1$ be this common subset of $P(x)$ and let $M_1'=\{\top(z): z\in M_1\}$. 
 Note that
 $|M_1'|\leq 2|\Gamma|\leq 2k$.  It can now be easily seen from the
 construction of $H_\Tf$ that for each $a'\succeq_T a$, we have
 $$\{b\in V(T)\colon b\preceq_T a\}\cap N^\uparrow(a')\subseteq M_0'\cup M_1'\cup \{a^\uparrow,\top(x),\top(x^\uparrow)\}.$$
 Since the set on the right hand side has size at most $836h+2k+4$, the claim follows.
\end{proofEnd}

Our next goal is to implement the combinatorial analysis described in
the previous section using first-order formulas working over $H_\Tf$.
Before we do this, let us see how the information about elements of
$U$ can be recovered from $H_\Tf$. Suppose $u\in U$ is a vertex for
which we know that $\pi(u)=a$ and $\chi(u)=c$. Then $\varpi(u)$ can be
easily inferred as $\top(x(a))$. Similarly, the color
$\kappa_{\Tf}(u,\varpi(u))$ can be obtained by applying
$\rho(\pth_T(a,\top(x(a))))$ to $c$. This in particular gives the
value of $\gamma(u)$. Finally, whenever for some ancestor~$y$ of
$x=\varpi(u)$, the value of $\varrho(\pth_Y(x,y))$ is stored in
$H_\Tf$, then the color $\kappa_{\Tf}(u,\top(y))$ can be obtained by
applying $\varrho(\pth_Y(x,y))$ to $\kappa_{\Tf}(u,\varpi(u))$. This
may happen when $y=\wg_\gamma(x)$ for some $\gamma\in \Gamma$, or when
$y\in \wL(x)$.

\bigskip

We are now ready to provide the promised implementation.

\begin{Alemma}
\label{cl:decode-splendid}
Fix $c_0,c_1\in [k]$. Then there are
formulas
$$\varphi_{c_0,c_1}(p_0,p_1),\quad
\psi_{c_0,c_1}(p_0,p_1),\quad\textrm{and}\quad
\{\,\zeta_{c_0,c_1,d_0,d_1}(p_0,p_1,q,q_0,q_1)\ \colon\ d_0,d_1\in
[k]\,\}$$ in the vocabulary of $H_\Tf$ such that the following holds
for all distinct $u_0,u_1\in U$ satisfying $\chi(u_0)=c_0$ and
$\chi(u_1)=c_1$, where $a_0=\pi(u_0)$ and $a_1=\pi(u_1)$.
 \begin{itemize}
 \item If $H_\Tf\models \varphi_{c_0,c_1}(a_0,a_1)$, then $u_0$ and
   $u_1$ are adjacent in $G$ if and only if
   $H_\Tf\models \psi_{c_0,c_1}(a_0,a_1)$.
 \item If $H_\Tf\not\models \varphi_{c_0,c_1}(a_0,a_1)$, then there is
   a unique $5$-tuple $(d_0,d_1,t,t_0,t_1)\in [k]^2\times V(T)^3$ such
   that $H_\Tf\models \zeta_{c_0,c_1,d_0,d_1}(a_0,a_1,t,t_0,t_1)$:
   \begin{itemize}
   \item $t=\top(\varpi(u_0)\wedge_Y\varpi(u_1))$;
   \item $t_0$ is the $\preceq_T$-maximum node of
     $\varpi(u_0)\wedge_Y\varpi(u_1)$ satisfying
     $t_0\preceq_T \pi(u_0)$;
   \item $t_1$ is the $\preceq_T$-maximum node of
     $\varpi(u_0)\wedge_Y\varpi(u_1)$ satisfying
     $t_1\preceq_T \pi(u_1)$;
   \item $d_0=\kappa_{\Tf}(u_0,t_0)$; and
   \item $d_1=\kappa_{\Tf}(u_1,t_1)$.
   \end{itemize}
 \end{itemize}
\end{Alemma}
\begin{proofEnd}
  We explain how, given $a_0,a_1\in V(T)$, $c_0,c_1\in [k]$, and
  access to the information present in $H_{\Tf}$, to either determine
  whether $u_0$ and $u_1$ are adjacent in $G$ or not, or find the
  $5$-tuple $(d_0,d_1,t,t_0,t_1)$ descibed in the statement. It is
  straightforward to encode the explained mechanism in first-order
  logic, which gives rise to the postulated first-order formulas.
 
 
  Let us adopt the notation from the previous section for $u_0$ and
  $u_1$.  In particular, $u_0$ is a $\gamma_0$-vertex, $u_1$ is a
  $\gamma_1$-vertex, $\varpi(u_0)=x_0$, $\varpi(u_1)=x_1$, and
  $z=x_0 \wedge_Y x_1$. As argued, $\gamma_0,\gamma_1,x_0,x_1$ can be
  inferred from $c_0,c_1,a_0,a_1$ given access to $H_\Tf$.
 
  As the first step, we find the $\preceq_Y$-maximal element of
  $\wL(x_0)\cap \wL(x_1)$. Call it $\tilde{z}$.  First, we consider
  the corner case when $x_0=x_1=\tilde{z}$. Then we have:
  \begin{itemize}[nosep]
  \item $t=\top(x_0)=\top(x_1)$;
  \item $t_0=a_0$;
  \item $t_1=a_1$;
  \item $d_0=c_0$; and
  \item $d_1=c_1$.
 \end{itemize}
 
 Second, we check whether both $\wL(x_0)$ and $\wL(x_1)$ contain a
 child of $\tilde{z}$.  Suppose for a moment that this is the case,
 and let $z'_0$ and $z'_1$ be these children, respectively.  Then by
 the maximality of $\tilde{z}$, we must have $z'_0\neq z'_1$, implying
 $z=\tilde{z}$.  It follows that:
 \begin{itemize}[nosep]
 \item $t=\top(z)$;
 \item $t_0$ is the parent in $T$ of $\top(z'_0)$;
 \item $t_1$ is the parent in $T$ of $\top(z'_1)$;
 \item
   $d_0=\rho(e(\top(z'_0)))\left(\kappa_{\Tf}(u_0,\top(z'_0))\right)$;
   and
 \item
   $d_1=\rho(e(\top(z'_1)))\left(\kappa_{\Tf}(u_1,\top(z'_1))\right)$.
 \end{itemize}
 As we argued, these values can be retrieved from $H_\Tf$ given
 $c_0,c_1,a_0,a_1$.
 
 Next, we consider a mix of the two cases above: $x_0=\tilde{z}$ and
 $\tilde{z}$ has a child $z'_1$ that belongs to $\wL(x_1)$.  Then
 again we have $z=\tilde{z}$ and:
 \begin{itemize}[nosep]
  \item $t=\top(z)$;
  \item $t_0=a_0$;
  \item $t_1$ is the parent in $T$ of $\top(z'_1)$;
  \item $d_0=c_0$; and
  \item $d_1=\rho(e(\top(z'_1)))\left(\kappa_{\Tf}(u_1,\top(z'_1))\right)$.
 \end{itemize}
 The case when $x_1=\tilde{z}$ and $\tilde{z}$ has a child $z'_0$ that
 belongs to $\wL(x_0)$ is symmetric.
 
 We claim that the four cases considered above cover all the
 situations when assumption~\eqref{p:lca-not-discovered} is not
 satisfied, that is, when $\LL(x_0)\cap Z\neq \emptyset$ and
 $\LL(x_1)\cap Z\neq \emptyset$.  Indeed, if this is the case, then
 $\wL(x_0)$ and $\wL(x_1)$ both contain $z$.  Moreover, $\wL(x_0)$
 contains the child of $z$ that is an ancestor of $x_0$, if existent,
 and similarly $\wL(x_1)$ contains the child of $z$ that is an
 ancestor of $x_1$, if existent. Then $z=\tilde{z}$ and in either way,
 one of the four cases considered above applies.
 
 Hence, from now on we proceed under the assumption
 that~\eqref{p:lca-not-discovered} holds.  Consequently, all the
 claims presented in the previous section can be applied.
 
 Denoting $\Pp(x_0)=\{A_1,\ldots,A_p\}$ and
 $\Pp(x_1)=\{B_1,\ldots,B_q\}$, we find the largest index $i$ such
 that $\top(A_i)=\top(B_i)$.  Note that $i$ and the kinds to which
 blocks $A_i$ and $B_i$ belong can be retrieved using the information
 stored along with sets $\wL(x_0)$ and $\wL(x_1)$.
 
 By \cref{lem:thecases},
 none of the blocks $A_i$ or $B_i$ can be fully mixed.  If either
 $A_i$ or $B_i$ is not biased, we may use
 \cref{lem:thecases}-\eqref{enum:A_not_biased} and
 \cref{lem:thecases}-\eqref{enum:B_not_biased} to directly infer
 whether $u_0$ and $u_1$ are adjacent in $G$ or not.  We are left with
 the case when both $A_i$ and $B_i$ are biased.  By
 \cref{lem:thecases}-\eqref{enum:both_biased}, 
 they are either both first-biased, or both second-biased.
 
 Suppose that both $A_i$ and $B_i$ are first-biased. Then, by 
 \cref{lem:thecases},
 we have:
 \begin{itemize}[nosep]
 \item $z=g_{\gamma_0}(x_1)$;
 \item $z_0=h_{\gamma_0}(z)\neq\bot$;
 \item if $z_0'$ is the parent in $Y$ of $z_0$, then $t_0$ is the
   parent in $T$ of $\top(z_0')$; and
 \item if $d_0'$ is the unique element of $\varrho(e(z_0))(\gamma_0)$,
   then $d_0=\rho(e(\top(z_0')))(d_0')$.
 \end{itemize}
 Here, the fact that $\varrho(e(z_0))(\gamma_0)$ consists of exactly
 one element of $\gamma_0$ is implied by the fact that $\Tf/\Pp$ is
 splendid, as asserted by \cref{lem:class}. It remains to retrieve
 $t_1$ and $d_1$.  For this, by~
 \cref{lem:thecases} we observe that if $\wg_{\gamma_0}(x_1)=\bot$
 then $x_1=z$ and we have
 \begin{itemize}[nosep]
  \item $t_1=a_1$ and
  \item $d_1=c_1$.
 \end{itemize}
 Otherwise, if $\wg_{\gamma_0}(x_1)\neq \bot$, then
 $\wg_{\gamma_0}(x_1)$ is the ancestor of $x_1$ that is a child of $z$
 and we have:
 \begin{itemize}[nosep]
 \item $t_1$ is the parent in $T$ of $\top(\wg_{\gamma_0}(x_1))$ and
 \item
   $d_1=\rho(e(\wg_{\gamma_0}(x_1)))\left(\kappa_{\Tf}(u_1,\top(\wg_{\gamma_0}(x_1)))\right)$.
 \end{itemize}
 
 The case when both $A_i$ and $B_i$ are second-biased is symmetric.
 As in all the cases we have either concluded whether $u_0$ and $u_1$
 are adjacent or not, or we have determined the $5$-tuple
 $(d_0,d_1,t,t_0,t_1)$, this finishes the proof.
\end{proofEnd}


\subsubsection{Shallow case}

We now treat the case when the quotient tree $(Y,U,\varrho,\varpi)$ is
shallow; recall that this means that $Y$ has height $1$.  As in the
previous section, we encode $\Tf$ in a structure $H_\Tf$ whose
universe is $V(T)$.  We encode the following information in $H_\Tf$:
\begin{itemize}[nosep]
\item the parent function of the tree $T$;
\item the mapping $a\mapsto \rho(e(a))$, where $a$ is a node of $T$
  and $e(a)$ is the edge of $T$ connecting $a$ with its parent;
\item the mapping $a\mapsto \top(x(a))$, where $a$ is a node of $T$
  and $x(a)$ is the node of $Y$ such that $a\in x(a)$; and
\item the mapping $a\mapsto \rho(\pth_T(a,\top(x(a))))$.
\end{itemize}
For $a\in V(T)$ we define $N^\uparrow(a)$ as before: $N^\uparrow(a)$
comprises all strict ancestors of $a$ in $T$ that are bound to~$a$ via
functions present in $H_\Tf$. We have the following analogue of
\cref{cl:SReach-splendid}.

\begin{lemma}\label{cl:SReach-shallow}
 For each $a\in V(T)$, 
 $$\Bigl|\{b\in V(T)\colon b\preceq_T a\}\cap \bigcup_{a' \succeq_T a} N^\uparrow(a')\Bigr|\leq 2.$$
\end{lemma}
\begin{proof}
  The only nodes that may be contained in the involved set are
  $a^\uparrow$ and $\top(x(a))$.
\end{proof}

We may also prove the following analogue of \cref{cl:decode-splendid}.

\begin{Alemma}
\label{cl:decode-shallow}
Fix $c_0,c_1\in [k]$. Then there
formulas
$$\{\zeta_{c_0,c_1,d_1,d_2}(p_0,p_1,q,q_0,q_1)\colon d_0,d_1\in
[k]\}$$ in the vocabulary of $H_\Tf$ such that the following holds for
all distinct $u_0,u_1\in U$ satisfying $c_0=\chi(u_0)$ and
$c_1=\chi(u_1)$, where $a_0=\pi(u_0)$ and $a_1=\pi(u_1)$. 
There is a 
unique $5$-tuple $(d_0,d_1,t,t_0,t_1)\in [k]^2\times V(T)^3$ such that
$H_\Tf\models \zeta_{c_0,c_1,d_0,d_1}(a_0,a_1,t,t_0,t_1)$:
\begin{itemize}
\item $t=\top(\varpi(u_0)\wedge_Y\varpi(u_1))$;
\item $t_0$ is the $\preceq_T$-maximum node of
  $\varpi(u_0)\wedge_Y\varpi(u_1)$ satisfying $t_0\preceq_T \pi(u_0)$;
\item $t_1$ is the $\preceq_T$-maximum node of
  $\varpi(u_0)\wedge_Y\varpi(u_1)$ satisfying $t_1\preceq_T \pi(u_1)$;
\item $d_0=\kappa_{\Tf}(u_0,t_0)$; and
\item $d_1=\kappa_{\Tf}(u_1,t_1)$.
\end{itemize}
\end{Alemma}
\begin{proofEnd}
  As in the proof of \cref{cl:decode-splendid}, we describe a
  mechanism of determining $(d_0,d_1,t,t_0,t_1)$ from
  $c_0,c_1,a_0,a_1$ given access to $H_\Tf$.  It is straightforward to
  formulate this mechanism in first-order logic, which gives rise to
  the postulated formulas.
 
  Let $x_0=\varpi(u_0)$ and $x_1=\varpi(u_1)$; note that $x_0$ and
  $x_1$ can be inferred from $a_0$ and $a_1$.  First, we check whether
  $x_0=x_1$. If this is the case, then we have
 \begin{itemize}[nosep]
  \item $t=\top(x_0)=\top(x_1)$;
  \item $t_0=a_0$;
  \item $t_1=a_1$;
  \item $d_0=c_0$; and
  \item $d_1=c_1$.
 \end{itemize}
 Otherwise, $x_0\wedge_Y x_1$ is equal to the root $r$ of $Y$. 
 Then:
 \begin{itemize}[nosep]
  \item $t=\top(r)$ is the root of $T$;
  \item $t_0=a_0$ if $x_0=r$, or $t_0$ is the parent of $\top(x_0)$ in $T$ otherwise;
  \item $t_1=a_1$ if $x_1=r$, or $t_1$ is the parent of $\top(x_0)$ in $T$ otherwise;
  \item $d_0=c_0$ if $x_0=r$, or
    $d_0=\rho(e(\top(x_0)))\circ \rho(\pth_T(a_0,\top(x_0)))(c_0)$
    otherwise; and
  \item $d_1=c_1$ if $x_1=r$, or
    $d_1=\rho(e(\top(x_1)))\circ \rho(\pth_T(a_1,\top(x_1)))(c_1)$
    otherwise.
 \end{itemize}
 This concludes the proof.
\end{proofEnd}

\subsubsection{Completing the induction}

We now utilize the understanding obtained in the previous sections to
complete the proof of \cref{thm:main} through an induction scheme.
Let $\ell \leq 3k^k$ be the length of the sequence of classes provided
by \cref{lem:factorization}.

Recall that we work with a $k$-NLC-tree
$\Tf=(T,U,\rho,\pi, \eta, \chi)$ generating $G$. We define a sequence
of factorizations $\Qq_1,\ldots,\Qq_\ell$ of $\Tf$ though backward
induction as follows:
\begin{itemize}[nosep]
\item $\Qq_\ell$ consists of one factor, being the whole tree $T$
  itself; and
\item for $i<\ell$, $\Qq_i$ is obtained from $\Qq_{i+1}$ by replacing
  each factor $F\in \Qq_i$ with all the factors of~$\Pp(\Tf_F)$.
\end{itemize}
Thus, \cref{lem:factorization} asserts that $\Qq_1$ is a factorization
of $\Tf$ into single-node factors.

Next, for each $i\in [\ell]$ and factor $F\in \Qq_i$ we define a
structure $J_F$. Intuitively, $J_F$ encodes the structure~$H_{\Tf_F}$
that we defined in the previous section, as well as all the structures
$J_{F'}$ for $F'\in \Pp(\Tf_F)$, constructed in the previous step of
the induction. Thus, the universe of $\Tf_F$ is $V(F)$, while the
relations in~$\Tf_F$ are defined by induction on $i$ as follows.

For $i=1$, the tree $F$ has exactly one node, say $a$. Structure $J_F$
stores only the value $\eta(a)$, encoded using unary relations on $a$.

\newcommand{\rt}{\mathsf{root}}

For $i>1$, the structure $J_F$ is constructed as a superposition of
the structure $H_{\Tf_F}$ and structures~$J_{F'}$ for
$F'\in \Pp(\Tf_F)$ as follows. First, consider the induced
$k$-NLC-tree $\Tf_F$ and construct the structure~$H_{\Tf_F}$ for it as
in the previous section. This structure has $V(F)$ as its
universe. Next, for each factor \mbox{$F'\in \Pp(\Tf_F)$}, consider
the structure $J_{F'}$ constructed in the previous step of induction
and add all the tuples from all the relations of $J_{F'}$ to $J_F$.
While doing this, we reuse relation names: we assume that all the
structures $J_{F'}$ are over the same vocabulary, so to obtain a
relation $R$ from this vocabulary in $J_F$ we take the union of
relations $R$ taken from structures $J_{F'}$ for $F'\in
\Pp(\Tf_F)$. Note here that the universes of structures $J_{F'}$ are
pairwise disjoint, and the vocabulary used for encoding $H_{\Tf_F}$ is
assumed to be disjoint from the vocabulary used for encoding
structures $J_{F'}$.  Finally, for technical reasons we add to $J_F$ a
function $\rt_i(\cdot)$ that maps each node $a\in V(F)$ to the root of
$F$.

Let now $J_{\Tf}\coloneqq J_T$, where $T$ is the unique factor of
$\Qq_\ell$. Further, let $J_{\Tf}^\star$ be the structure obtained
from~$J_\Tf$ by adding $U$ to the universe, together with unary and
binary relations encoding mappings $u\mapsto \pi(u)$ and
$u\mapsto \chi(u)$, for $u\in U$.

First, we verify that $J_{\Tf}^\star$ contains all the information
needed to reconstruct $G$.

\begin{Alemma}
\label{cl:edge-relation}
There is a first-order formula $\alpha(p_0,p_1)$ over the vocabulary
of $J_{\Tf}^\star$ such that for all $u_0,u_1\in U$, we have
$J_{\Tf}^\star\models \alpha(u_0,u_1)$ if and only if
$u_0u_1\in E(G)$.
\end{Alemma}
\begin{proofEnd}
  For a pair of vertices $a_0,a_1\in V(T)$, let the {\em{level}} of
  $(a_0,a_1)$ be the smallest integer $i$ such that $a_0$ and $a_1$
  belong to the same factor of $\Qq_i$.  As $\Qq_\ell$ consists of one
  factor --- the whole tree $T$ --- the level of every pair is upper
  bounded by $\ell$.  We shall inductively define formulas
  $\beta^i_{c_0,c_1}(p_0,p_1)$ for $c_0,c_1\in [k]$ and $i\in [\ell]$
  satisfying the following property: for every pair
  $(a_0,a_1)\in V(T)^2$ of level at most $i$, if there are vertices
  $u_0,u_1\in U$ satisfying $\pi(u_0)=a_0$, $\pi(u_1)=a_1$,
  $\chi(u_0)=c_0$, and $\chi(u_1)=c_1$, then
  $J_{\Tf}^\star\models \beta^i_{c_0,c_1}(a_0,a_1)$ iff
  $u_0u_1\in E(G)$. If we succeed in this, then formula
  $\alpha(u_0,u_1)$ can be written by first defining $a_0=\pi(u_0)$,
  $a_1=\pi(u_1)$, $c_0=\chi(u_0)$, and $c_1=\chi(u_1)$, and then
  applying $\beta^{\ell}_{c_0,c_1}(a_0,a_1)$.
 
  Consider first the base case $i=1$. As factorization $\Qq_1$ places
  every node of $T$ in a different factor, then condition that
  $(a_0,a_1)$ has level at most $1$ boils down to $a_0=a_1$. Hence
  $\beta^1_{c_0,c_1}(a_0,a_1)$ only needs to check that $a_0=a_1$ and
  that $(c_0,c_1)\in \eta(a_0)$.
 
  We proceed to the induction step. Let $F$ be the factor of $\Qq_i$
  that contains both $a_0$ and $a_1$. We shall assume that the
  quotient tree $\Tf_F/\Pp(\Tf_F)$ is splendid, hence we will use
  formulas provided by \cref{cl:decode-shallow} for the $k$-NLC-tree
  $\Tf_F$. Note here that the structure $H_{\Tf_F}$ encoding $\Tf_F$
  is contained in $J_{\Tf}^\star$. Hence, these formulas may be
  applied in $J_{\Tf}^\star$ in the same manner as in~$H_{\Tf_F}$,
  provided that we appropriately relativize them to the elements of
  $V(F)$; these can be distinguished as elements mapped to the root of
  $F$ by $\rt_i(\cdot)$.  The reasoning in the other case, when
  $\Tf_F/\Pp(\Tf_F)$ is shallow, proceeds in the same way and is even
  simpler, as we may use \cref{cl:decode-shallow} instead of
  \cref{cl:decode-splendid}.
 
  We first check whether $\varphi_{c_0,c_1}(a_0,a_1)$ holds in
  $H_{\Tf_F}$. If this is the case, then we may immediately determine
  whether $u_0$ and $u_1$ are adjacent in $G$ by checking whether
  $\psi_{c_0,c_1}(a_0,a_1)$ holds in $H_{\Tf_F}$. Otherwise, using
  formulas $\zeta_{c_0,c_1,d_0,d_1}(p_0,p_1,q,q_0,q_1)$ we can find
  suitable colors $d_0,d_1\in [k]$ and nodes $t,t_0,t_1\in V(F)$, as
  described in \cref{cl:decode-splendid}. Note here that if $F'$ is
  the factor of $\Pp(\Tf_F)$ that contains the least common ancestor
  of $a_0$ and $a_1$, then
  \begin{itemize}[nosep]
  \item $t=\top(F')$;
  \item $t_0=\pi_{F'}(u_0)$;
  \item $t_1=\pi_{F'}(u_1)$;
  \item $d_0=\chi_{F'}(u_0)$; and
  \item $d_1=\chi_{F'}(u_1)$.
 \end{itemize}
 Hence, to decide whether $u_0u_1\in E(G)$, it suffices to check
 whether $J_{\Tf}^\star\models \alpha^{i-1}_{d_0,d_1}(t_0,t_1)$, which
 is a formula that we constructed in the previous step of induction.
\end{proofEnd}

\newcommand{\Gaif}{\mathsf{Gaif}}

Recall that the {\em{Gaifman graph}} of a structure $A$ is the
undirected graph $\Gaif(A)$ whose vertex set is the universe of $A$,
and where two elements are considered adjacent if and only if they
appear simultaneously in a tuple in a relation in $A$.  Define
$$\Dd\coloneqq \{\,\Gaif(J_{\Tf}^\star)\ \colon\ \Tf\textrm{ is a }k\textrm{-NLC-tree generating a graph from }\Cc\ \}.$$


That the class $\Dd$ has bounded treewidth is then proved using the
characterization of treewidth through the {\em{strong reachability}}
relation, with the help of \cref{cl:SReach-splendid} and \cref{cl:SReach-shallow}.

  For the proof of \cref{cl:bnd-tw}, we need several definitions.

  Let $G$ be a graph and let $\leq$ be a {\em{vertex ordering}} of
  $G$, that is, a linear order on the vertex set of $G$. For a vertice
  $u$ and $v$ of $G$, we say that $v$ is {\em{strongly reachable}}
  from $u$ in $\leq$ if $v\leq u$ and in $G$ there exists a path $P$
  from $u$ to $v$ such that $u<w$ for every internal vertex $w$ of
  $P$. Then, we define the {\em{strong reachability set}} of $u$,
  denoted $\SReach_{\infty}[G,\leq,u]$ as the set of all vertices of
  $G$ that are strongly reachable from $u$ in $\leq$. The {\em{strong
      $\infty$-coloring number}} of $G$ is defined as
$$\scol_\infty(G)=\min_{\leq}\max_{u\in V(G)}|\SReach_{\infty}[G,\leq,u]|,$$
where the minimum ranges over all vertex orderings of $G$. It is
folklore that the strong $\infty$-coloring number essentially
coincides with treewidth.

\begin{theorem}[see e.g. Chapter~1, Theorem~1.19 of~\cite{notes}]\label{thm:scol-tw}
  For every graph $G$, the treewidth of $G$ is equal to
  $\scol_{\infty}(G)-1$.
\end{theorem}

We now use Theorem~\ref{thm:scol-tw} together with
\cref{cl:SReach-splendid} and \cref{cl:SReach-shallow} to prove the
following.
\begin{Alemma}
\label{cl:bnd-tw}
For every graph $G\in \Dd$, the treewidth of $G$ is at most
$3k^k\cdot (836h+2k+4)$.
\end{Alemma}
\begin{proofEnd}
  By Theorem~\ref{thm:scol-tw}, it suffices to give a vertex ordering
  of $G$ where each strong reachability set has size at most
  $3k^k\cdot (836h+2k+4)+1$.  Let $G=\Gaif(J_{\Tf}^\star)$, where
  $\Tf=(T,U,\rho,\pi, \eta, \chi)$ is a $k$-NLC-tree that generates a
  graph from $\Cc$. Then $V(G)=U\cup V(T)$. Let $\leq$ be a vertex
  ordering of $G$ constructed as follows: first put all the nodes of
  $T$ in any order that extends $\preceq_T$ (that is, $u\preceq_T v$
  entails $u\leq v$), and then put all the vertices of $U$ in any
  order. Our goal is to establish an upper bound on the sizes of
  strong reachability sets with respect to the ordering $\leq$.
 
  Observe that for $u\in U$, we have
  $\SReach_\infty[G,\leq,u]=\{u,\pi(u)\}$, so this is a set of size
  $2$. Consider then any $a\in V(T)$. From the construction of $J_\Tf$
  it follows that all the edges of $G$ which connect two nodes $V(T)$
  in fact connect a node of $T$ with its ancestor. Hence, we have
 $$\SReach_\infty[G,\leq,a]\subseteq \{a\}\cup \bigcup_{i=1}^\ell N^\uparrow_i(a),$$
 where $N^\uparrow_i(a)$ is the set $N^\uparrow(a)$ evaluated in the structure $H_{\Tf_{F_i}}$, where $F_i$ is the factor from $\Qq_i$ that contains $a$. By \cref{cl:SReach-splendid} and \cref{cl:SReach-shallow}, each of the sets $N^\uparrow_i(a)$ has size at most $836h+2k+4$, so $$|\SReach_\infty[G,\leq,a]|\leq \ell\cdot (836h+2k+4)+1=3k^k\cdot (836h+2k+4)+1,$$ as required.
\end{proofEnd}

The bound obtained in \cref{cl:bnd-tw} is not optimal, and could be easily reduced. Note that it is not known whether there is a collapse in the hierarchy of classes with bounded treewidth with respect to first-order transductions, that is, whether there exist integers $k<k'$ with the property that the class of graphs with treewidth at most $k'$ can be transduced from the class of graphs with treewidth at most $k$. We conjecture that this is not the case.

\newcommand{\wM}{\widehat{M}}
\newcommand{\wG}{\widehat{G}}
\newcommand{\wDd}{\widehat{\Dd}}

\medskip
We are now able to prove \cref{thm:main}, which we restate below.

\mainthm*
\begin{proof}
  For a graph $G$, let $\wG$ be the graph obtained from $G$ by
  subdividing every edge $uv$ twice, that is, replacing it with a path
  $u-s^{u}_{uv}-s^{v}_{uv}-v$. Let $\wDd=\{\wG\colon G\in \Dd\}$. As
  subdividing edges does not increase the treewidth and $\Dd$ has
  bounded treewidth by \cref{cl:bnd-tw}, the same bound also applies
  to $\wDd$.

  We now prove that there is a transduction from $\wDd$ onto $\Cc$,
  hence establishing the only non-trivial implication of the theorem.

  Consider any graph $G\in \Cc$. Let $\Tf$ be any $k$-NLC-tree that
  generates $G$.  Let $M=\Gaif(J_{\Tf}^\star)$. We argue that $G$ can
  be transduced from $\wM\in \wDd$ using a fixed transduction that
  depends only on $k$.

  We first argue that the structure $J^\star_\Tf$ can be transduced
  from $\wM$.  First, we add colors to distinguish the original
  vertices of $M$ from the subdividing vertices (i.e. vertices
  $s^{u}_{uv}$ and $s^{v}_{uv}$ introduced when constructing~$\wM$
  from $M$). Now, recall that the vocabulary of $J^\star_\Tf$ consists
  only of unary relations and partial functions. Unary relations
  present in $J^\star_\Tf$ can be introduced directly. For every
  partial function~$f$ present in~$J^\star_\Tf$, we transduce it as
  follows. First, we introduce a unary predicate $Z_f$ which selects
  vertices~$s^u_{u\,f(u)}$ for $u$ ranging over the domain of
  $f$. Then it is straightforward to interpret $f$ using a first-order
  formula involving~$Z_f$. Thus, we have introduced all the relations
  present in $J^\star_\Tf$, and it remains to use a universe
  restriction formula to dispose of all the subdividing vertices,
  which should not be included in the universe of $J^\star_\Tf$.

  Now that $J^\star_\Tf$ has been transduced from $\wM$, we can use
  formula $\alpha(p_0,p_1)$ provided by \cref{cl:edge-relation} to
  interpret the edge relation of $G$ in $J^\star_\Tf$. Restricting the
  universe to $U$ finishes the construction of $G$ from $\wM$ by means
  of a transduction.
\end{proof}

Finally, let us discuss the algorithmic aspects of the proof. Given a graph $G\in \Cc$, we can compute a $k$-NLC-tree generating $G$ in cubic time~\cite{oum2006approximating}, for some constant $k$. The hierarchical factorization provided by \cref{lem:factorization} can be computed in polynomial time, because the result of Colcombet~\cite{colcombet2007combinatorial} is effective. It is straightforward forward to see that all the further elements of the construction, like determining the types, partitioning into blocks, etc., which amount to the construction of the structure~$J^\star_\Tf$, can be carried out in polynomial time. Thus, given $G\in \Cc$, we can in polynomial time compute a graph of bounded treewidth $H$ from which $G$ can be transduced, together with a suitable monadic extension of $H$. The interpretation yielding $G$ from this monadic extension of $H$ can be computed as~well.

\section{Some combinatorial consequences of Theorem~\ref{thm:main}}\label{sec:consequences}

\cref{thm:main} asserts that each class with bounded rankwidth and
stable edge relation is a transduction of a class with bounded
treewidth. We now derive some consequences of this result.

Classes with bounded treewidth are examples of classes with bounded
expansion \cite{nevsetvril2008grad}.  Recall that a class $\Cc$ has
{\em bounded expansion} if there exists a function
$f\colon\mathbb N\rightarrow\mathbb N$ with the property that every
graph~$H$ such that a subdivision of $H$ with edges subdivided at
most $r$ times is a subgraph of a graph in $\Cc$ has average degree
at most $f(r)$.  (The reader is referred to \cite{Sparsity} for an
in-depth study of these classes.)

These classes are characterized by the existence of special covers.
Let {\sf complexity} be a graph parameter, such as treewidth or
rankwidth. A class $\mathscr C$ has {\em low {\sf complexity} covers}
if for each positive integer $p$ there exists a constant $C_p$ and a
class $\mathscr X_p$ with bounded {\sf complexity}, such that each
graph $G\in\mathscr C$ can be covered by $C_p$ induced subgraphs
$H_1,\dots,H_{C_p}\in\mathscr X_p$ in such a way that every subset of~$p$ vertices of $G$ are jointly covered by some $H_i$
($1\leq i\leq C_p$).

\medskip
Recall that the {\em treedepth} of a graph $G$ \cite{Sparsity} is the
minimum number of levels of a rooted forest $Y$ such that $G$ is a
subgraph of the ancestor-descendant closure of $Y$. Equivalently, the treedepth of a graph
$G$ is the minimum clique number of a supergraph of $G$ that is a
trivially perfect graph.  The following result
follows from the characterization of bounded expansion in terms of low
treedepth colorings.

\begin{theorem}[\cite{nevsetvril2008grad}]
\label{thm:BEcov}
A class has bounded expansion if and only if it has low treedepth covers.
\end{theorem}

An extension of this result gives a characterization of the graph
classes that are transductions of classes with bounded expansion. Following~\cite{SBE_TOCL}, we say that such classes have \emph{structurally bounded expansion}.

\begin{theorem}[\cite{SBE_TOCL}]
\label{thm:SBE}
A class has structurally bounded
expansion if and only it has low shrubdepth covers.
\end{theorem}

Recall that a class $\mathscr S$ has {\em bounded shrubdepth} if there
exist constants $m$ and $h$ such that for every graph $G\in\mathscr S$
there is a rooted tree $Y$ with set of leaves $L(Y)=V(G)$, a coloring
$c:L(Y)\rightarrow [m]$ and an assignment $v\mapsto f_v$ of a
symmetric function $f_v\colon [m]\times[m]\rightarrow\{0,1\}$ to each
internal node $v$ of~$Y$, in such a way that two vertices
$u,v\in V(G)$ are adjacent in $G$ if and only if
$f_{u\wedge_Y v}(c(u),c(v))=1$~\cite{Ganian2012, Ganian2017}. In
particular, the subgraph of $G$ induced by each single color class is a cograph. Since cographs are perfect, in particular we have $\chi(G)\leq m\,\omega(G)$. We deduce the following corollary
of \cref{thm:SBE}.

\begin{corollary}
\label{cor:SBEchi}
For every structurally bounded expansion class $\Cc$ there exists a
constant $C$ such that the vertex set of every $G\in\Cc$ can be
partitioned into at most $C$ classes, each inducing a cograph.
	
In particular, every structurally bounded expansion class is linearly
$\chi$-bounded.
\end{corollary}

Note that a class has bounded shrubdepth if and only if it can be
transduced from a class with bounded treedepth \cite{Ganian2012}.

In an effort to generalize low treedepth coverings further, classes
with low rankwidth covers have been studied in \cite{KwonPS20}.  As a
direct consequence of \cref{thm:main} and \cref{cor:SBEchi}, we have:

\thmrwcovchi*
\begin{proof}
 Let $\Cc$ be the class in question.
 Taking $p=1$ in the definition, for every graph $G\in \Cc$, we can partition the vertex set of $G$ into a bounded number of parts, each of which induces a subgraph that belongs to a class $\Dd$ that has bounded rankwidth and a stable edge relation. By \cref{thm:main}, $\Dd$ can be transduced from a class of bounded treewidth, hence it has structurally bounded expansion. By \cref{cor:SBEchi} we conclude that $\Dd$ is linearly $\chi$-bounded, so it follows that $\Cc$ is linearly $\chi$-bounded as well.
\end{proof}

\newcommand{\Oh}{{\mathcal{O}}}

It is known that the chromatic number of graphs
with (linear) cliquewidth at most $k$ cannot be computed in
$f(k)n^{2^{o(k)}}$ time for any computable function $f$, unless {\sf ETH}
fails~\cite{golovach2018cliquewidth}. However, it follows from what
precedes that for each class $\Cc$ with bounded rankwidth and
stable edge relation there is an $\Oh(n^3)$-time algorithm, which gives
a constant factor approximation for the chromatic number. Indeed, given a graph $G$ from the considered class, we can first use the result of Oum and Seymour~\cite{oum2006approximating} to compute in cubic time a $k$-NLC-tree of $G$ for some constant $k$ (or any equivalent decomposition, such as a clique expression). Then, using standard dynamic programming we can compute the clique number of the graph in linear time. By \cref{thm:rwcovchi}, this clique number is a constant-factor approximation of the chromatic number.

We also deduce the following result.
\thmrwcov*
\begin{proof}
  If a class has structurally bounded expansion, then it has low
  shrubdepth covers~\cite{GajarskyHOLR16}, which are special instances of low rankwidth
  covers. Moreover, as bounded expansion classes are nowhere dense,
  they are monadically stable \cite{adler2014interpreting}, hence
  structurally bounded expansion classes have a stable edge relation.
	
  Conversely, assume a class $\Cc$ has low rankwidth covers and stable
  edge relation. Then for each integer~$p$ there exists a constant
  $C_p$ and a class $\mathscr R_p$ with bounded rankwidth such that
  each graph $G\in\Cc$ can be covered by $C_p$ induced subgraphs
  $H_1,\dots,H_{C_p}\in\mathscr R_p$ in such a way that every subset
  of~$p$ vertices of $G$ are jointly covered by some $H_i$
  ($1\leq i\leq C_p$). As $\Cc$ has a stable edge relation, it
  excludes some half-graph~$F$. Obviously, we can require that
  $\mathscr R_p$ contains only induced subgraphs of graphs in
  $\Cc$. Thus graphs in $\mathscr R_p$ exclude $F$ as well, so
  $\mathscr R_p$ has a stable edge relation. By \cref{thm:main}, $\mathscr R_p$ can be transduced from a class with bounded
  treewidth, hence $\mathscr R_p$ has structurally bounded expansion. It follows from
  \cref{thm:SBE} that there exists $C_p'$ and a class $\mathscr T_p$
  with bounded shrubdepth such that each graph $H_i$ can be covered by
  $C_p'$ induced subgraphs $T_{i,1},\dots,T_{i,C_p'}\in\mathscr T_p$
  in such a way that every subset of $p$ vertices of $H_i$ are jointly
  covered by some $T_{i,j}$. We deduce that $\Cc$ has low shrubdeth
  covers, so it has structurally bounded expansion.
\end{proof}

In \cite{SBE_TOCL}, it was stressed that one of the main difficulties
arising when considering low shrubdepth covers of structurally bounded
expansion classes (whose existence is asserted in \cref{thm:SBE}) is
that we do not know if they may be computed in polynomial time (and
that polynomial-time computation of these covers for $p=2$ ensures
that FO-model checking is FPT on the class). A consequence of this
paper is that for a class with structurally bounded treewidth (that is, a class with bounded rankwidth and stable edge relation), and for
each integer $p$, low shrubdepth covers with parameter $p$ can be
computed in polynomial time. Such a property also holds for
structurally bounded degree classes (that is, transductions of classes
with bounded degree)~\cite{GajarskyHOLR16}, as well as classes
obtained from bounded expansion classes by a transduction consisting a
bounded number of subgraph complementations
\cite{gajarsky2018recovering}. We conjecture that this holds in
general.

\begin{restatable}{conjecture}{conjcov}
  For every structurally bounded expansion class $\Cc$, computing a
  low shrubdepth cover of a graph $G\in\Cc$ at depth $p$ is fixed
  parameter tractable when parameterized by $p$. 
\end{restatable}



\section{Monadic dependence meets stability}
\label{sec:ms}

In this section we prove \cref{thm:ms}, which shows that the equivalence of the first three
conditions of \cref{thm:main} (and \cref{thm:lrw}) is in fact a more
general phenomenon that occurs in every monadically dependent graph
class.
In our proof, we shall need the following classical
theorem.
\begin{theorem}[Canonical Ramsey Theorem \cite{erdos1950combinatorial}]
\label{thm:canonical}
For every integer $n$ there exists an integer~$N$ with the following
property: Suppose that all pairs $(a,b)$ of integers with
$1\leq a<b\leq N$ are arbitrarily distributed into classes. Then there
is an increasing sequence of integers $1\leq x_1<x_2<\dots<x_n\leq N$
such that one of the following four sets of conditions holds, where it
is assumed that $1\leq \alpha<\beta\leq n$;
$1\leq \gamma<\delta\leq n$:
\begin{enumerate}[nosep]
\item All $(x_\alpha,x_\beta)$ belong to the same class.
\item $(x_\alpha,x_\beta)$ and $(x_\gamma,x_\delta)$ belong to the
  same class if, and only if, $\alpha=\gamma$.
\item $(x_\alpha,x_\beta)$ and $(x_\gamma,x_\delta)$ belong to the
  same class if, and only if, $\beta=\delta$.
\item $(x_\alpha,x_\beta)$ and $(x_\gamma,x_\delta)$ belong to the
  same class if, and only if, $\alpha=\gamma$; $\beta=\delta$.
\end{enumerate}
\end{theorem}

Let us now proceed to the proof of \cref{thm:ms}, restated below.
\msthm*

\newcommand{\N}{\mathbb{N}}

\begin{proof}
  Implications
  \enumref{enum:ms}$\Rightarrow$\enumref{enum:s}$\Rightarrow$\enumref{enum:hg}
  are obvious, so it remains to prove the following: if a class $\Cc$
  is monadically dependent but also monadically unstable, then in fact
  $\Cc$ has an unstable edge relation. Hence, assume that~$\Cc$ is
  monadically unstable.  In the following, we write $\binom{[n]}{2}$
  for the set of all pairs of integers $(i,j)$ such that
  $1\leq i<j\leq n$.
 
  A formula $\alpha(\tup x)$ is {\em{functional}} on a class if
  there is a variable $x\in \tup x$ such that for every $G$ in the class and
  $u\in V(G)$, there exists at most one tuple
  $\tup u\in V(G)^{\tup x}$ such that $G\models \alpha(\tup u)$ and
  $\tup u(x)=u$.  We shall say that a triple of formulas
  $\tau=(\alpha(\tup x),\beta(\tup y),\eta(\tup x,\tup y))$ in a
  monadic vocabulary of graphs is {\em{problematic}} if there exists a
  monadic expansion $\Cc^+$ of $\Cc$, whose vocabulary contains the
  vocabularies of $\alpha$, $\beta$, and $\eta$, such that~$\alpha$
  and $\beta$ are functional on $\Cc^+$, and for every $n\in \N$ there
  exists $G\in \Cc^+$ and tuples
  $\tup a_1,\ldots,\tup a_n\in V(G)^{\tup x}$ and
  $\tup b_1,\ldots,\tup b_n\in V(G)^{\tup y}$ satisfying the
  following:
 \begin{itemize}[nosep]
 \item for all $i\in [n]$ we have $G\models \alpha(\tup a_i)$ and
   $G\models \beta(\tup b_i)$; and
 \item for all $(i,j)\in \binom{[n]}{2}$ we have
   $G\models \eta(\tup a_i,\tup b_j)$ and
   $G\models \neg \eta(\tup a_j,\tup b_i)$.
 \end{itemize}
 Note that we do not specify whether $\eta(\tup a_i, \tup b_i)$ should
 hold or not in $G$.  The pair of sequences $\tup a_1,\ldots,\tup a_n$
 and $\tup b_1,\ldots,\tup b_n$ as above shall be called a
 {\em{$\tau$-ladder}} of length $n$ in $G$.  Observe that if in graphs
 from $\Cc^+$ one can find arbitrarily long $\tau$-ladders, then
 $\eta$ is unstable on $\Cc^+$.
 
 As $\Cc$ is monadically unstable, by \cref{thm:stable_trans} we know
 that there is a transduction from $\Cc$ onto the class of all finite
 half-graphs.  By the definition of a transduction, this implies that
 there exists a monadic expansion $\Cc^+$ of $\Cc$ and a formula
 $\phi(x,y)$ with two free variables $x$ and $y$ such that $\phi$ is
 unstable on $\Cc^+$.  By taking $\alpha(x)$ and $\beta(y)$ to be true
 formulas, we conclude the following.
 
 \begin{zclaim}
   There exists a problematic triple of formulas.
 \end{zclaim}

 We now investigate the properties of problematic formulas.

 \begin{zclaim}
   If $\tau=(\alpha(\tup x),\beta(\tup y),\eta(\tup x,\tup y))$ is
   problematic, then so is
   $\tau'=(\alpha(\tup x),\beta(\tup y),\neg \eta(\tup x,\tup y))$.
 \end{zclaim}
 \begin{claimproof}
   It suffices to observe that reversing both sequences in a
   $\tau$-ladder yields a $\tau'$-ladder.
 \end{claimproof}

 \begin{zclaim}
   If the triple
   $\tau=(\alpha(\tup x),\beta(\tup y),\eta_1(\tup x,\tup y)\vee
   \eta_2(\tup x,\tup y))$ is problematic, then at least one of the
   triples
   $\tau_1=(\alpha(\tup x),\beta(\tup y),\eta_1(\tup x,\tup y))$ and
   $\tau_2=(\alpha(\tup x),\beta(\tup y),\eta_2(\tup x,\tup y))$ is
   problematic.
 \end{zclaim}
 \begin{claimproof}
   By assumption, there is a monadic expansion $\Cc^+$ of $\Cc$ such
   that there are arbitrarily long $\tau$-ladders in graphs from
   $\Cc$.  Suppose $\tup a_1,\ldots,\tup a_n$ and
   $\tup b_1,\ldots,\tup b_n$ is such a $\tau$-ladder in some
   $G\in \Cc$.  Observe that for all $(i,j)\in \binom{[n]}{2}$, we
   have $G\models \eta_1(\tup a_i,\tup b_j)$ or
   $G\models \eta_2(\tup a_i,\tup b_j)$.  By Ramsey's theorem and
   since $n$ can be chosen arbitrarily large, by restricting attention
   to a sub-ladder we may assume that one of these cases holds for
   every pair $(i,j)\in \binom{[n]}{2}$, say the first one by
   symmetry.  However, for all $(i,j)\in \binom{[n]}{2}$ we also have
   $G\models \neg(\eta_1(\tup a_j,\tup b_i)\vee \eta_2(\tup a_j,\tup
   b_i))$, which implies $G\models \neg \eta_1(\tup a_j,\tup b_i)$.
   We conclude that $\tup a_1,\ldots,\tup a_n$ and
   $\tup b_1,\ldots,\tup b_n$ form a $\tau_1$-ladder of length $n$.
   As $n$ can be chosen arbitrarily large, $\tau_1$ is problematic.
 \end{claimproof}

%

 \begin{zclaim}
   If a triple
   $\tau=(\alpha(\tup x),\beta(\tup y),\eta(\tup x,\tup y))$ is
   problematic and
   $\eta(\tup x,\tup y)=\exists z\ \zeta(\tup x,\tup y,z)$, then there
   is a problematic triple of the form
   $\tau'=(\alpha'(\tup x'),\beta'(\tup y'),\zeta(\tup x',\tup y'))$
   where either $(\tup x',\tup y')=(\tup x\cup \{z\},\tup y)$ or
   $(\tup x',\tup y')=(\tup x,\tup y\cup \{z\})$.
 \end{zclaim}
 \begin{claimproof}
   Consider any $n\in \N$ and let $N$ be the integer given by the
   Canonical Ramsey Theorem (\cref{thm:canonical}) for $n$.  By
   assumption, there is a monadic expansion $\Cc^+$ of $\Cc$ such that
   there exist arbitrarily long $\tau$-ladders in graphs from $\Cc^+$.
   Hence, we can find a $\tau$-ladder
   $\tup a_1,\ldots,\tup a_{2N},\tup b_1,\ldots,\tup b_{2N}$ of 
   length~$2N$ in some $G\in \Cc^+$.  By restricting attention to a
   sub-ladder consisting of every odd element of the sequence
   $\tup a_1,\ldots,\tup a_{2N}$ and every even element of the
   sequence $\tup b_1,\ldots,\tup b_{2N}$, and appropriately
   reindexing, we find a $\tau$-ladder
   $\tup a_1,\ldots,\tup a_{N},\tup b_1,\ldots,\tup b_{N}$ of length-$N$ in $G$ such
   that $G\models \eta(\tup a_i,\tup b_i)$ for all $i\in [n]$.  Note
   that tuples $\tup a_1,\ldots,\tup a_N$ have to be pairwise
   different, because for each $i\in [N]$, the smallest $j\in [N]$
   satisfying $G\models \eta(\tup a_i,\tup b_j)$ is equal to
   $i$. Similarly, tuples $\tup b_1,\ldots,\tup b_N$ have to be
   pairwise different as well.
  
   Let $x\in \tup x$ and $y\in \tup y$ be the variables witnessing
   that $\alpha$ and $\beta$ are functional, respectively.  For
   $i\in [n]$, let $a_i=\tup a_i(x)$ and $b_i=\tup b_j(y)$. As
   $\alpha$ is functional, we conclude that vertices $a_1,\ldots,a_N$
   are pairwise different, and similarly vertices $b_1,\ldots,b_N$ are
   pairwise different as well.  Let $A=\{a_1,\ldots,a_N\}$ and
   $B=\{b_1,\ldots,b_N\}$, and let $G^{AB}$ be a monadic expansion of
   $G$ where $A$ and $B$ are additionally distinguished using unary
   predicates, which we shall respectively call $A$ and $B$ by a
   slight abuse of notation.
  
   Let $\prec$ be the (strict) lexicographic order on
   $\binom{[N]}{2}$.  Observe that there exists a formula
   $\lambda(\tup x,\tup y,\tup x^\circ,\tup y^\circ)$, where
   $\tup x^\circ$ and $\tup y^\circ$ are copies of $\tup x$ and
   $\tup y$, respectively, such that the following holds: if
   $(\tup a,\tup b)=(\tup a_i,\tup b_j)$ and
   $(\tup a^\circ,\tup b^\circ)=(\tup a_{i^\circ},\tup b_{j^\circ})$
   for some $(i,j),(i^\circ,j^\circ)\in \binom{[n]}{2}$, then
   $G^+\models \lambda(\tup a,\tup b,\tup a^\circ,\tup b^\circ)$ if
   and only if $(i,j)\prec (i^\circ,j^\circ)$.  Indeed, the formula
   $\forall\tup w\ \left[(B(w)\wedge \beta(\tup w)\wedge \eta(\tup
     a^\circ,\tup w))\rightarrow \eta(\tup a,\tup w)\right]$ (where
   the variable $w\in \tup w$ corresponds to the variable
   $y\in \tup y$) allows us to check the assertion $i\leq i^\circ$.  A
   formula expressing $j\leq j^\circ$ can be written in a symmetric
   way.  Then the condition $(i,j)\prec (i^\circ,j^\circ)$ can be
   expressed using a boolean combination of assertions
   $i\leq i^\circ$, $i\geq i^\circ$, $j\leq j^\circ$, and
   $j\geq j^\circ$.
  
   As for every pair $(i,j)\in \binom{[N]}{2}$ we have
   $G\models \exists z\ \zeta(\tup a_i,\tup b_j,z)$, there is a vertex
   $c\in V(G)$ such that $G\models \zeta(\tup a_i,\tup b_j,c)$.  Let
   $C$ be an inclusion-wise minimal subset of $V(G)$ such that for
   each $(i,j)\in \binom{[N]}{2}$ there exists $c\in C$ satisfying
   $G\models \zeta(\tup a_i,\tup b_j,c)$.  For every $c\in C$,
   define
   $$J(c)=\left\{\,(i,j)\in \binom{[N]}{2}\ \colon\ \zeta(\tup
     a_i,\tup b_j,c)\,\right\}.$$ Note that by the minimality of $C$,
   the sets $J(c)$ are pairwise not contained in one another.  Let
   $G^{ABC}$ be the monadic expansion of $G^{AB}$ where $C$ is
   additionally distinguished using a unary predicate $C$.
  
  Now, for $c,c'\in C$, we set
  $$c\sqsubset c'\qquad \textrm{if and only if}\qquad \min_{\prec}\,\left(J(c)\setminus J(c')\right)\ \prec\ \min_{\prec}\,\left(J(c')\setminus J(c)\right).$$
  It is straightforward to see that $\sqsubset$ is a (strict) linear
  order on $C$.  Let us partition pairs $(i,j)\in \binom{[N]}{2}$ into
  classes $\{I(c)\colon c\in C\}$ as follows:
  $$(i,j)\in I(c) \qquad \textrm{if and only if}\qquad c = \min_{\sqsubset}\,\{\,d\in C\ \colon\ (i,j)\in J(d)\,\}.$$
  Using the formula $\lambda$ we can easily write a formula
  $\kappa(\tup x,\tup y,z)$ with the following property: for all
  $(i,j)\in \binom{N}{2}$ and $c\in C$, we have
  $G^{ABC}\models \kappa(\tup a_i,\tup b_j,c)$ if and only if
  $(i,j)\in I(c)$.
  
  By the Canonical Ramsey Theorem (\cref{thm:canonical}) there exists
  $F\subseteq [N]$ such that $|F|=n$ and one of the following
  conditions is satisfied:
 \begin{enumerate}[(1)]
 \item all pairs $(i,j)\in\binom{F}{2}$ belong to the same class
   $I(c)$, for some $c\in C$;
 \item there exist pairwise different $c_i$ such that
   $(i,j)\in I(c_i)$ for all $(i,j)\in \binom{F}{2}$;
 \item there exist pairwise different $c_j$ such that
   $(i,j)\in I(c_j)$ for all $(i,j)\in \binom{F}{2}$;
 \item there exist pairwise different $c_{i,j}$ such that
   $(i,j)\in I(c_{i,j})$ for all $(i,j)\in \binom{F}{2}$.
 \end{enumerate}
 Let $G^{A'B'C}$ be the monadic expansion of $G^{ABC}$ where sets
 $A'=\{a_i\colon i\in F\}$ and $B'=\{b_i\colon i\in F\}$ are
 additionally distinguished using unary predicates $A'$ and $B'$.
  
 We first consider the second case above. Let $G^{A'B'CD}$ be a
 monadic expansion of $G^{A'B'C}$ that distinguishes the single vertex
 $b_{\max F}$ using a unary predicate $D$.  Consider the formula
  $$\alpha'(\tup x,z)=A'(x)\wedge C(z)\wedge \alpha(\tup x)\wedge
  \exists {\tup y}\ \left[D(y)\wedge \beta(\tup y)\wedge \kappa(\tup
    x,\tup y,z)\right].$$ Observe that for any
  $\tup u\in V(G)^{\tup x}$ and $w\in V(G)$, we have
  $G^{A'B'CD}\models \gamma(\tup u,w)$ if and only if
  $\tup u=\tup a_i$ for some $i\in F$ and $w=c_i$.  As
  $\alpha(\tup x)$ is functional, it follows that so is
  $\alpha'(\tup x,z)$.  It is now straightforward to see that
  $\{(\tup a_i,c_i)\colon i\in F\}$ and $\{\tup b_i\colon i\in F\}$
  form a $\tau'$-ladder in $G^{A'B'CD}$ of length $n$, where
  \mbox{$\tau'=(\alpha'(\tup x,z),\beta(\tup y),\zeta(\tup x,\tup
    y,z))$}.  Hence, if the second case occurs for infinitely many
  $n$, then $\tau'$ is problematic.

  The same argument applies if the first case occurs for infinitely
  many $n$, and a symmetric argument applies when the third case
  occurs for infinitely many $n$.  We are left with considering the
  situation where the fourth case occurs for infinitely many $n$.  Let
  $S=\{c_{i,j}\colon (i,j)\in \binom{F}{2}\}$. Observe that if we
  choose any subset $P\subseteq S$ and distinguish it using a unary
  predicate $P$ in a monadic expansion $G^{A'B'CP}$ of $G^{A'B'C}$,
  then the formula
  $$\xi(\tup x,\tup y)=A'(x)\wedge \alpha(\tup x)\wedge B'(y)\wedge
  \beta(\tup y)\wedge \exists z\ \left[P(z)\wedge \kappa(\tup x,\tup
    y,z)\right],$$ is true exactly for those tuples $\tup a_i$ and
  $\tup b_j$ for which $(i,j)\in \binom{F}{2}$ and $c_{i,j}\in P$.
  Hence, using $\xi$ and different choices of $P$ we may interpret in
  graphs $G^{A'B'CP}$ all subgraphs of a half-graph of order $n$.  It
  follows that there is a transduction from $\Cc$ onto the class of
  all bipartite graphs; this contradicts the assumption that $\Cc$ is
  monadically dependent.
 \end{claimproof}
 
 By the above claims we infer that there is a problematic triple
 $(\alpha(\tup x),\beta(\tup y),\eta(\tup x,\tup y))$ such that
 $\eta(\tup x,\tup y)$ is an atomic formula.  In particular, this
 means that there is a monadic expansion $\Cc^+$ of $\Cc$ such that
 $\eta$ is unstable on $\Cc^+$.  Since $\eta$ is atomic, it is of one
 of the following forms: a unary predicate applied to any variable;
 the equality relation applied to any pair of variables; or the edge
 relation $E(\cdot,\cdot)$ applied to any pair of variables.  The
 first two cases cannot happen, as such formulas are stable on every
 class of graphs.  We conclude that the last case occurs, hence $\Cc$
 has an unstable edge relation.
\end{proof}

\section{Conclusion and Perspectives}\label{sec:conclusion}

We have started to explore the theory of monadic dependence and
monadic stability from a graph theoretical point of view. Several
interesting questions and conjectures arise from our studies.  To put
our research in perspective, we show in \cref{fig:lattice_ext} the
following extended semi-lattice of property inclusions.

\begin{figure}[h!t]
  \centering
  \includegraphics[width=0.65\textwidth]{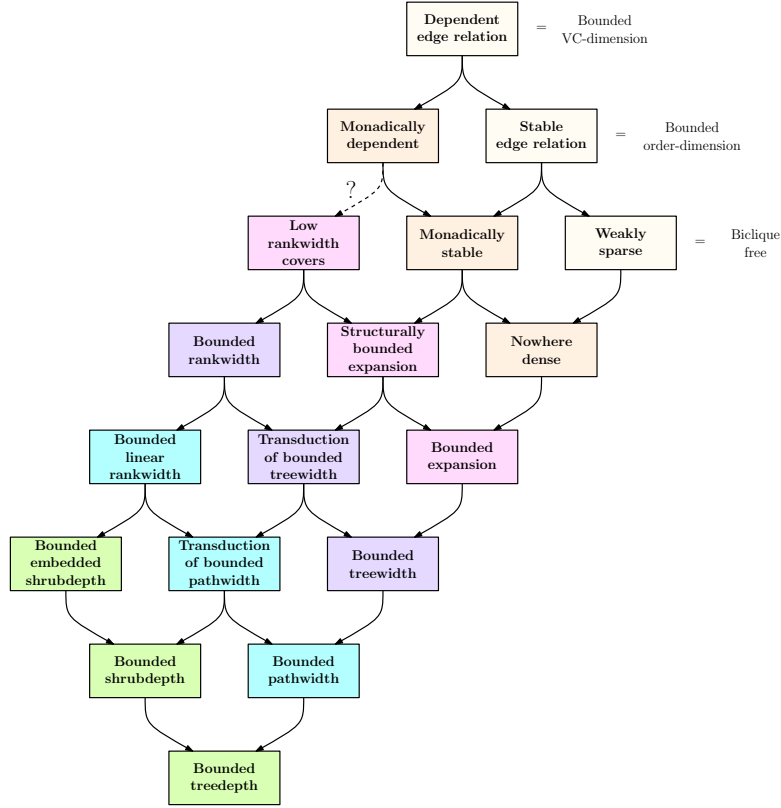}
  \caption{The extended semi-lattice of property inclusions. }
  \label{fig:lattice_ext}
\end{figure}

A quick examination of the figure reveals an unresolved question of
prime importance. While \cref{thm:lrw} and \cref{thm:main} exactly
identify classes of structurally bounded pathwidth/treewidth as
monadically stable classes that have bounded (linear) rankwidth, the
chart does not specify the alignment of {\em structurally nowhere
  dense} classes (i.e.\ transductions of nowhere dense
classes). Clearly, every structurally nowhere dense class of graphs is
monadically stable, but the precise relationship between these notions
remains to be understood. It would be even consistent with our
knowledge if the two concepts coincided for classes of graphs.  If
this was true, it would reveal very strong structural qualities of
monadically stable classes of graphs, which could be used in the
algorithmic context.

\begin{conjecture}
  A graph class is monadically stable if and only if it is
  structurally nowhere dense.
\end{conjecture}


\medskip
Obviously, besides classes of bounded pathwidth or treewidth, there
are multiple other notions of sparsity whose structural analogs could
be investigated. For instance, can we characterize {\em{structurally
    planar}} classes, that is, images of the class of planar graphs
under transductions? More generally, one may consider images under
transductions of classes with forbidden minors or with forbidden
topological minors.  So far, suitable characterizations have been
given for classes with {\em structurally bounded
  degree}~\cite{GajarskyHOLR16} and with {\em structurally bounded
  expansion}~\cite{SBE_TOCL}.  Such characterizations, if efficiently
constructive, are very helpful in the design of fixed-parameter
algorithms for the FO model-checking problem, as was done in the case
of classes with structurally bounded degree~\cite{GajarskyHOLR16}.
Based on the understanding revealed
in~\cite{GajarskyHOLR16,SBE_TOCL}, we hypothesize that such
characterizations may rely on the concept of {\em{covers}} (see
\cref{sec:consequences}). For instance, transductions of classes with
bounded expansion are characterized by the existence of such covers
(see \cref{thm:SBE}). This motivates the following:

\begin{conjecture}
  Every class with low rankwidth covers is monadically dependent.
\end{conjecture}

Finally, we recall the conjecture we posed in \cref{sec:consequences}.

\conjcov*



\bibliography{ref}
\clearpage
\appendix
\newcommand{\inappendix}{yes!}
\end{document}